\documentclass[10pt, conference, letterpaper]{IEEEtran}

\IEEEoverridecommandlockouts

\usepackage{cite}
\usepackage{comment}
\usepackage{graphicx,subfigure}
\usepackage{boxedminipage}
\usepackage{alltt}
\usepackage{multicol}
\usepackage{amsmath,amssymb,amsfonts}
\usepackage{xcolor}
\usepackage{wrapfig}
\usepackage{mathptmx}
\usepackage{epstopdf}
\usepackage{todonotes}
\usepackage{caption}
\usepackage{cuted}
\usepackage{listings}
\usepackage[linesnumbered,algoruled,boxed,lined]{algorithm2e}
\usepackage{amsthm}
\SetKwProg{MyStruct}{Struct}{ contains}{end}
\usepackage{graphicx}
\usepackage{capt-of}
\usepackage{textcomp}
\usepackage{xcolor}

\newtheorem{lemma}{Lemma}
 
\newtheorem{definition}{Definition}

\newcommand{\fangyu}[1]{\textcolor{black}{#1}}

\newcommand{\chen}[1]{\textcolor{black}{#1}}
\newcommand{\niu}[1]{\textcolor{black}{#1}}
\usepackage{microtype}

\def\BibTeX{{\rm B\kern-.05em{\sc i\kern-.025em b}\kern-.08em
    T\kern-.1667em\lower.7ex\hbox{E}\kern-.125emX}}
\begin{document}

\title{Fast-HotStuff: A Fast and Robust BFT Protocol for Blockchains
}


\author{
\IEEEauthorblockN{Mohammad M. Jalalzai\IEEEauthorrefmark{1}\IEEEauthorrefmark{2}, Jianyu Niu\IEEEauthorrefmark{1}\IEEEauthorrefmark{2}, Chen Feng\IEEEauthorrefmark{1}\IEEEauthorrefmark{2} 
and 
Fangyu Gai\IEEEauthorrefmark{1}\IEEEauthorrefmark{2}
} 

\IEEEauthorblockA{ \IEEEauthorrefmark{1}School of Engineering, The University of British Columbia, Kelowna, Canada }

\IEEEauthorblockA{ \IEEEauthorrefmark{2}Blockchain@UBC, The University of British Columbia, Vancouver, Canada }

\IEEEauthorblockA{ 
{\{m.jalalzai, jianyu.niu, chen.feng, fangyu.gai\}@ubc.ca} }
}

\maketitle
\thispagestyle{plain}
\pagestyle{plain}
\begin{abstract}
The HotStuff protocol is a recent breakthrough in  Byzantine  Fault  Tolerant  (BFT)  consensus that enjoys both responsiveness and linear view change by creatively adding a round to classic two-round BFT protocols like PBFT. 
Despite its great advantages, HotStuff has a few limitations. First, the additional round of communication during normal cases results in higher latency. Second, HotStuff is vulnerable to certain performance attacks, which can significantly deteriorate its throughput and latency.
To address these limitations, we propose a new two-round BFT protocol called Fast-HotStuff,  which enjoys responsiveness and efficient view change that is comparable to the linear view-change in terms of performance. Our Fast-HotStuff has lower latency and is more robust against the performance attacks that HotStuff is susceptible to.


\end{abstract}


\section{Introduction}\label{Section:Introduction}
Byzantine Fault Tolerant (BFT) consensus has received considerable attention in the last decade due to its promising application in blockchains. 
Classic BFT-based protocols suffer from the scalability issue due to the high-cost \textcolor{black}{when the leader is replaced\footnote{A leader is also called primary. A primary proposes a value to the network. The network tries to reach a consensus to serially execute the proposal.}.}
For example, in PBFT,  $O(n^3)$ authenticators/signatures are exchanged and processed during a \textcolor{black}{ leader replacement} (where $n$ corresponds to the network size) \cite{Castro:1999:PBF:296806.296824}. More recent protocols like SBFT \cite{SBFT}, Zyzzyva \cite{Kotla:2008:ZSB:1400214.1400236} and BFT-Smart \cite{BFT-SMART} have reduced the number of signatures transmitted and processed within the network to $O(n^2)$. Still, the cost of view change is high for a large $n$, causing an enormous delay. This becomes an even bigger problem when a rotating primary\footnote{View is a number that is deterministically mapped to the ID of the primary/leader. Hence, when a view is changed, the primary is changed in the protocol.}
protocol is used where the primary is replaced after each block proposal.

Several state-of-the-art BFT protocols, including Tendermint~\cite{tendermint} and Casper FFG \cite{Casper}, have been proposed to address the issue of expensive view change. These protocols not only support linear message complexity by using advanced cryptography (such as aggregated or threshold signatures), but also enable frequent leader rotation by adopting the chain structure (which is popular in blockchains).
Moreover, these protocols can be pipelined, which can further improve their performances, and meanwhile, make them much simpler to \chen{implement}. These protocols operate with a synchronous core, and there is a pessimistic bound $\Delta$ on the network delay. Hence, the protocol has $O(\Delta)$ latency (instead of the actual network latency), therefore lacking responsiveness. 



\fangyu{The HotStuff protocol is the first to achieve both linear view change\footnote{In HotStuff \cite{Hot-stuff}, the linear view change has been defined as the linear number of signatures/authenticators sent over the wire during consensus.} and responsiveness, solving a decades-long open problem in BFT consensus. 
Linear view change enables fast leader rotation while responsiveness drives the protocol to consensus at the speed of wire $O(\delta)$, where $\delta$ is the actual network latency. Both are desirable properties in the blockchain space. }
In BFT protocols, a decision is made after going through several phases, and each phase usually takes one round of communication before moving to the next.
HotStuff introduces a chained structure borrowed from blockchain to pipeline all the phases into a unifying propose-vote\footnote{A primary proposes the next block once it receives $n-f$ votes for the previous block without waiting for the previous block to get committed.} pattern which significantly simplifies the protocol and improves protocol throughput.
Perhaps for this reason, pipelined HotStuff (also called chained HotStuff) has been adopted by Facebook's DiemBFT \cite{libra-BFT} (previously known as LibraBFT), Flow platform~\cite{Hentschel2020FlowSC}, as well as Cypherium Blockchain\cite{cypherium}. Throughout the paper, we refer to the pipelined version of HotStuff (not the basic HotStuff).

To achieve both properties (responsiveness and linear view change), HotStuff uses threshold signatures and creatively adopts a three-chain commit rule. In the three-chain commit rule, it takes two uninterrupted rounds of communication plus one another round (not necessarily uninterrupted) among replicas to commit a block. For example, for a block at round $r$ to get committed, two consecutive rounds $r+1$, $r+2$ and one another round $r+k$, where $k>2$ must be completed successfully. 
This is in contrast with the two-chain commit rule commonly used in most other BFT protocols~\cite{Castro:1999:PBF:296806.296824, tendermint, Pala, Casper, jalalzai2020hermes}.

In this paper, we mainly focus on the chained (pipelined) version of HotStuff with \textcolor{black}{a} rotating primary. In the rotating primary mechanism, a dedicated primary is assigned each time to propose a block. The rotating primary mechanism provides two additional important properties to a partially synchronous protocol (that cannot be achieved through the stable primary mechanism, where the primary is changed only upon failure). First, by randomly selecting the rotating primary nodes, the protocol limits the time window for malicious adversaries to perform a Denial-of-Service attack on the primary node \cite{libra-BFT}. Secondly, by rotating primary nodes, each node gets an equal opportunity to propose a block. This can be a very important property when nodes receive rewards by proposing blocks. \textcolor{black}{Moreover, pipelining significantly improves the protocol throughput.}


HotStuff has made trade-offs to achieve linear view change and responsiveness. First, it adds a round of consensus to the classic two-round BFT consensus. Secondly, 
to achieve higher throughput, a rotating primary in pipelined HotStuff proposes a block without waiting for its child block to get committed (proposals are pipelined). 

This results in the formation of forks, which can be exploited by Byzantine primaries to overwrite blocks from correct primaries before they are committed. 
It should be noted that classic BFT protocols \cite{Castro:1999:PBF:296806.296824,Jalal-Window,SBFT} do not allow fork formation. The main reason for \textcolor{black}{the} generation of forks is that a proposed block is not yet committed while the next block is being proposed by the following primary. Furthermore, forking also breaks the requirement for
three consecutive round\textcolor{black}{s} to commit a block, resulting in additional delay. As shown in  Figure \ref{fig:Simple-Forking-Example}, the last three blocks in replica $i$'s chain have been added during rounds $r$, $r+1$\textcolor{black}{,} and $r+2$ (where the block during $r+2$ is the latest block in replica $i$'s chain). Now, a Byzantine primary during the round $r+3$ proposes a block that points to the grandparent of the latest block as its parent. In this case, the block proposed during the round $r+3$ points to the block proposed during the round $r$ as its parent block instead of pointing to the block of the round $r+2$. This is acceptable according to  HotStuff's voting rule. A replica accepts a proposal and votes for it as long as it does not point to the predecessor of the grandparent of the latest block. Hence,  allowing \textcolor{black}{the} formation of forks in HotStuff results in lower throughput and higher latency.  Byzantine replicas will break the consecutive order of blocks added to the chain through forking. As it can be seen in Figure \ref{fig:Simple-Forking-Example} initially the blocks were added to the chain through consecutive rounds ($r$, $r+1$ and $r+2$). But after forking, the consecutive rounds in the chain \textcolor{black}{are} broken by the block in the round $r+3$. As a result, the new block sequence in the chain is $r$ and $r+3$. Since $r+3$ is the most recent block, \textcolor{black}{the} next primary will extend the block proposed during $r+3$. Therefore, averted transactions have to be re-proposed. Furthermore,  new consecutive rounds have to be completed for the block in round $r$, to get committed. This causes \textcolor{black}{an} increase in latency and throughput degradation. \textcolor{black}{It should be noted that forking attacks in HotStuff can be prevented if a stable primary mechanism (instead of rotating) is used. But in that case, the protocol will lose its ability to limit the window of time in which an adversary can launch a Denial-of-service attack (when primary nodes are selected randomly). Moreover, by not using rotating primary, each node will not be able to get an equal chance of adding a block to the chain.}
 \begin{figure}
    \centering
    \includegraphics[width=8cm,height=4cm]{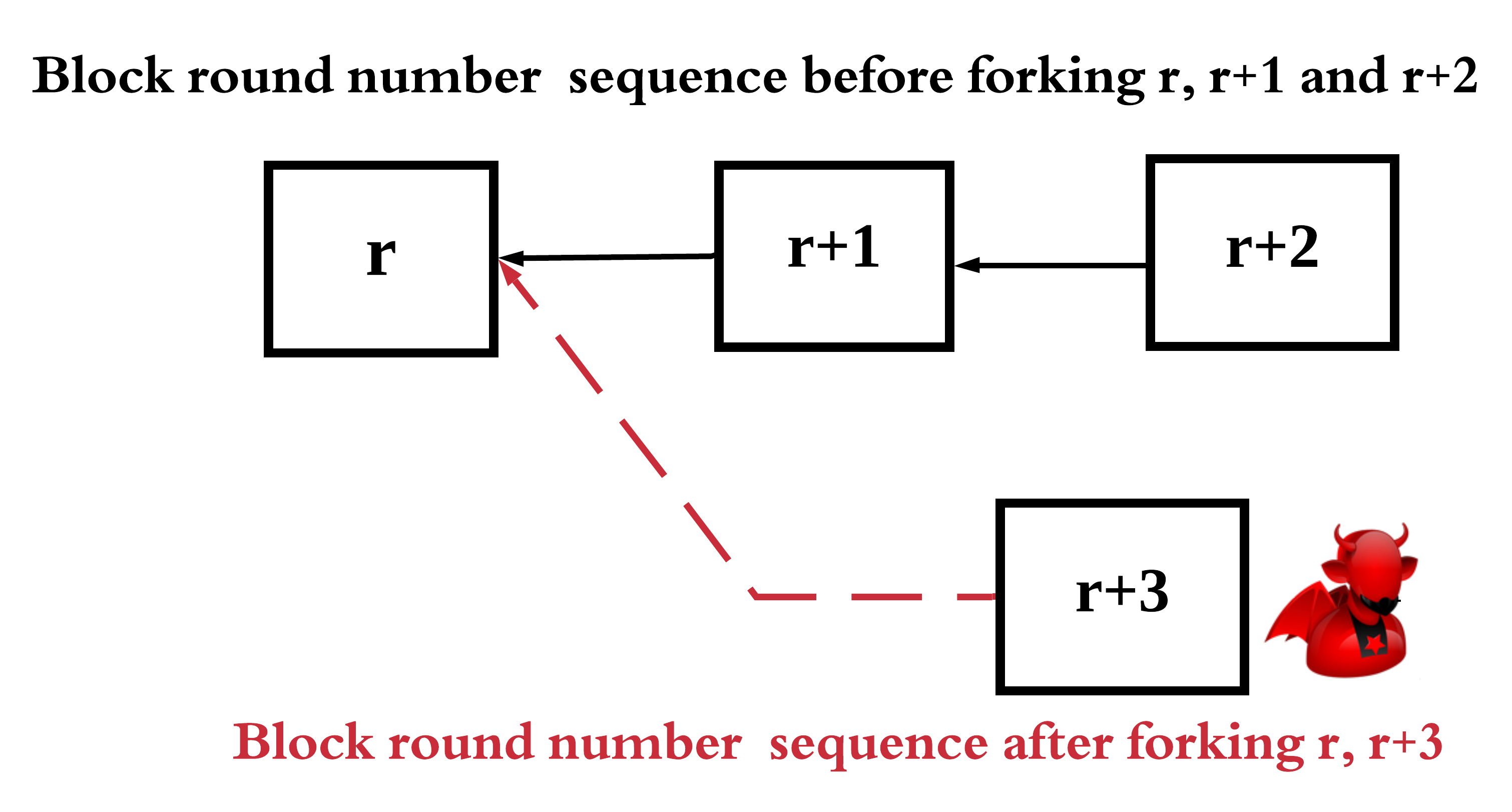}
    \caption{\textbf{The primary of the round $r+3$ performs the forking attack.}}
    \label{fig:Simple-Forking-Example}
\end{figure}

Therefore, a natural question arises here that:

\emph{ Is it possible to design a consensus protocol that avoids the trade-offs made by HotStuff while still \textcolor{black}{having} responsiveness and efficient view change \textcolor{black}{(with rotating primary mechanism)}?}

To answer the above question, we designed Fast-HotStuff. Fast-HotStuff is robust against forking attack and has lower latency in comparison to the HotStuff. Moreover, similar to the HotStuff, Fast-HotStuff provides efficient view change and responsiveness.  

Forking attack cannot be performed in Fast-HotStuff. This is \textcolor{black}{because} in Fast-HotStuff a primary has to provide proof that the proposed block is extending the latest block seen by the majority of replicas. Based on Figure \ref{fig:Simple-Forking-Example}, the primary of the round $r+3$ has to provide a proof that the parent of 
its proposing block is the latest block seen by \textcolor{black}{the} majority of replicas. Any replica will not vote for a block unless it satisfies the requirement for the proof of extending the latest block seen by the majority of replicas. This prevents Byzantine primaries to propose a block that points to an older block as its parent, as done in Figure \ref{fig:Simple-Forking-Example}. As a result, fork formation is avoided.

Similarly, in the absence of failure after two rounds of communication, a replica in Fast-HotStuff commit a block without waiting for the third round. The protocol guarantees that if a single replica commits a block, other replicas will eventually commit the same block at the same height. This guarantee can be achieved in HotStuff in the absence of failure after three rounds of communication.



There are mainly two types of view change in Fast-HotStuff, $1)$ Happy path, in which a block is successfully proposed by a primary and $n-f$ votes are collected for it by the next primary in line and $2)$ The unhappy path in which the primary failed to propose a block, majority of replicas timeout and move to the next view (primary).
In Fast-HotStuff in the absence of primary failure, during normal rotation or happy path of the primary (view change) no overhead is needed. This means when no primary failure occurs, then the primary rotation or the view change
requires only one authenticator to send and \textcolor{black}{verify}. Therefore, the happy path in Fast-HotStuff is linear. 
Unfortunately, we are not able to avoid the transfer of quadratic view change messages over the wire during view change due to primary failure (also called unhappy path). But we were able to reduce the number of signatures to be processed (verified) by each replica by the factor of $O(n)$. The number of authenticators (aggregated signatures) to be verified in Fast-HotStuff during \textcolor{black}{the} unhappy path is reduced to only two. This is in contrast with other two-chain responsive BFT protocols \cite{SBFT,Kotla:2008:ZSB:1400214.1400236,BFT-SMART,Prime,proteus1}
where at least $O(n)$ authenticators need to be verified by each replica ($O(n^2)$ for $n$ replicas).
Therefore, reducing the number of signatures to be verified significantly improves performance for large $n$.
The tradeoff for improvements achieved in Fast-HotStuff is a small overhead in the block each time a primary fails (unhappy path). For example, this overhead is at most $\approx 1.4\%$ of the block size (1MB) and  $\approx 0.7\%$ of $2MB$ block for the network size of $100$ replicas. We also show in Sections   \ref{Section: Design of Fast-HotStuff} and \ref{Section: Evaluation} that this overhead does not have a significant impact on performance. 

Overall, during the unhappy path, the view change is optimized by reducing the number of authenticators to be verified by each replica. Whereas during happy path the protocol enjoys linear view change similar to HotStuff.

The rest of the paper is organized as follows.   In Section \ref{Section: System Model}, \chen{we} describe the system model.  In Section \ref{Section: HotStuff overview}, we present an  overview of HotStuff protocol. 
Section \ref{Section:Tradeoffs}, discuss in detail the tradeoffs made by HotStuff to achieve linear view change and responsiveness.
Section \ref{Section: Design of Fast-HotStuff} provides algorithms for basic Fast-HotStuff and pipelined Fast-HotStuff protocols. Section \ref{Section: Proof-of-correctness} provide safety and liveness proofs for the basic and pipelined Fast-HotStuff. 
Evaluation and related work are presented in Sections \ref{Section: Evaluation} and \ref{Section: Related Work}, respectively. The paper is concluded in Section \ref{Section:Conclusion}.

\section{System Model and Preliminaries}\label{Section: System Model} 
\subsection{System Model}
We consider a system with $n = 3f+1$ parties (also called replicas) denoted by the set $N$ in which at most $f$ replicas are Byzantine. 
Byzantine replicas may behave \textcolor{black}{arbitrarily}, whereas correct (or honest) replicas always follow the protocol.
We assume a \emph{partial synchrony model} presented in \cite{Dwork:1988:CPP:42282.42283}, where there is a known bound $\Delta$ on message transmission delay (when the network is in the synchronous mode). This $\Delta$ bound holds after an unknown asynchronous period called \textit{Global Stabilization Time} (GST). \textcolor{black}{In practice, the system can only make progress if the $\Delta$ bound on message delivery remains for a sufficiently long time. Furthermore, assuming the $\Delta$ bound is everlasting simplifies the discussion.}
All exchanged messages are signed.
Adversaries are computationally bound and cannot forge signatures or message digests (hashes) but with negligible probability. \textcolor{black}{Moreover, Fast-HotStuff uses a static adversary model. In the static adversary model,  nodes are chosen to be corrupted at the beginning of the protocol.}

\subsection{Preliminaries}
\noindent \textbf{View and View Number.}
During each view, a dedicated primary is responsible for proposing a block. Each view is identified by a monotonically increasing number called view number. Each replica uses a deterministic function to translate the view number into a replica $ID$, which will act as primary for that view. Therefore, the primary (or leader) of each view is known to all replicas. 

\noindent \textbf{Signature Aggregation.} Fast-HotStuff uses signature aggregation \cite{Short-Signatures-from-the-Weil-Pairing,Boneh:2003,BDNSignatureScheme} to obtain a single collective signature of constant size instead of appending all replica signatures when a primary fails. 
As the primary $p$ receives the message $M_i$ with their respective signatures $\sigma_i \gets sign_i(M_i)$ from each replica $i$, the primary then uses these received signatures to generate an aggregated signature $\sigma \gets AggSign(\{M_i, \sigma_i\}_{i \in N})$. The aggregated signature can be verified by replicas given the messages $M_1,M_2,\ldots, M_y$, where $ 2f+1 \leq y \leq n$, the aggregated signature $\sigma$, and public keys $PK_1,PK_2,\ldots,PK_y$. To authenticate message senders as done in previous BFT-based protocols \cite{Castro:1999:PBF:296806.296824,Lamport:1982:BGP:357172.357176, Jalal-Window} each replica $i$ keeps the public keys of other replicas in the network.
Fast-HotStuff can use any signature scheme where message sender identities are known. HotStuff uses threshold signature \cite{Cachin:2005:Practical:Asynch, Short-Signatures-from-the-Weil-Pairing, Practical-Threshold-Signatures} where identities of message senders are not known. We also show in practice, that aggregated signature  used by Fast-HotStuff and threshold signatures used by HotStuff have comparable performance.

\noindent \textbf{Quorum Certificate (QC) and Aggregated QC.}
A block $B$'s quorum certificate (QC) is proof that more than $2n/3$ nodes (out of $n$) have voted for this block. A  QC 
 comprises an aggregated signature or threshold (from $n-f$ signatures from distinct replicas) built by signing block hash in a specific view. 
A block is certified when its QC is received, and certified blocks' freshness is ranked by their view numbers.
In particular, we refer to a certified block with the highest view number that a node knows as the \emph{latest/highest} certified block. The latest $QC$ a node knows is called $highQC$ (in HotStuff). \textcolor{black}{ Fast-HotStuff considers a global view for $highQC$ to avoid the additional round. Therefore, $highQC$ in Fast-HotStuff is the latest $QC$ held by the majority of replicas or the $QC$ with the higher view than the latest $QC$ held by the majority of replicas. The block proposal includes the $highQC$ as well as the proof of $highQC$ within itself. }
It should be noted that unlike classic BFT in HotStuff and Fast-HotStuff the $QC$ is also used to point to the parent block.  An aggregated $QC$ or $AggQC$ is simply a vector built from a concatenation of $n-f$ or $2f+1$ $QCs$.

\noindent \textbf{Block and Block Tree.} Clients send transactions to primary, who then batch transactions into blocks.  A block also has a field that it uses to point to its parent. In the case of HotStuff and Fast-HotStuff, a $QC$ (which is built from $n-f$ votes) is used as a pointer to the parent block.
Every block except the genesis block must specify its parent block and include a QC for the parent block. In this way, blocks are chained together. 
As there may be forks, each replica maintains a block tree (referred to as $blockTree$) of received blocks. 
 Two blocks $B_{v+1}$ and $B_{v+2}$ are conflicting if  neither $B_{v+1}$ is predecessor \textcolor{black}{n}or ancestor to the block $B_{v+2}$, nor  $B_{v+2}$ is predecessor or ancestor to the block $B_{v+1}$. But conflicting blocks
 have a common predecessor block ($B_{v}$ in this case).
 The parent block $B_{v}$ is the vertex of the fork, as shown in Figure \ref{fig:Conflicting blocks}.

\noindent \textbf{Chain and Direct Chain.} If a block $B$ is being added over the top of a block $B'$ (such that $B.parent=B'$), then these two blocks make one-chain. If another block $B^*$ is added over the top of the block $B$, then $B'$, $B$, and $B^*$ make two-chain and so on.
There are two ways that the chain or tree of blocks grows. First, the chain grows in a continuous manner where the chain is made between two consecutive blocks. For example, for two blocks $B$ and $B'$ we have $B.curView = B'.curView+1$ and $B.parent=B'$. Thus, a direct chain exists between block $B$ and $B'$. This is also called one-direct chain. If $B.curView = B'.curView+1$ and $B.parent=B'$ and  $B^*.curView = B.curView+1$ and $B^*.parent=B$, then we can say a two-direct chain exists among blocks $B'$, $B$ and $B^*$.
But there is also the possibility that one or more views between two views fail to generate block due to primary being Byzantine or network failure. In that case $B.curView = B'.curView+k$ where $k>1$ and $B.parent=B'$ and direct chain does not exist between $B'$ and $B$.

 \begin{figure}
    \centering
    \includegraphics[width=5.5cm,height=2cm]{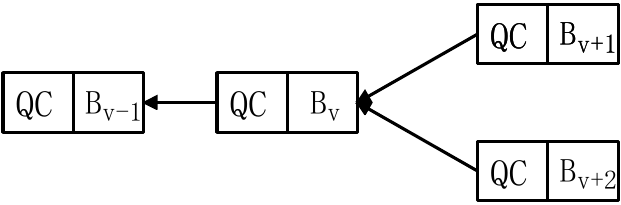}
    \caption{\textbf{A simple case of conflicting blocks}. Block $B_v$ is a fork vertex, and blocks $B_{v+1}$ and $B_{v+2}$ are two conflicting blocks.}
    \label{fig:Conflicting blocks}
\end{figure}

\section{HotStuff in a nutshell}
\label{Section: HotStuff overview}
\niu{In this section, we provide a brief description of the three-chain HotStuff protocol.} 
There are two variants of the three-chain HotStuff protocol, the basic and pipelined HotStuff. Since pipelined HotStuff is mainly used due to its higher throughput, here, we will focus on pipelined HotStuff. 
We describe the pipelined HotStuff operation beginning in the view $v$. At the beginning of the view $v$ a dedicated primary is selected by replicas operating in the view number $v$. The primary proposes a block $B_v$ (multicast  to all replicas) that extends the block certified by the $highQC$ it has seen. Upon receipt of the first proposed block from the primary, each replica sends its vote (a signature on the block) to the primary of view $v$. Upon receipt of $n-f$ votes, the primary of the view $v$ builds a $QC$ and forwards the $QC$ to the primary of the view $v+1$ \footnote{If the primary of the view $v$ is Byzantine, then it may not share the $QC$ for the view $v$ with the primary of the view $v+1$.}.  The $QC$ of the view $v$ is the $highQC$ that the primary of the view $v+1$ is holding. Therefore, the primary in the view $v+1$ will propose its block ($B_{v+1}$) with the $QC$ from the view $v$, as shown in the Figure \ref{fig:Pipelined-HotStuff}.

Every replica has to keep track of two local parameters in HotStuff $1$) $last\_voted\_view$, the latest view when the replica has voted and $2$) $last\_locked\_view$, the view of the grandparent block of the $last\_voted\_view$. As it can be seen from the Figure \ref{fig:Pipelined-HotStuff}, when a replica  votes for the block $B_{v+2}$ during view $v+2$, its $last\_voted\_view$ will be $v+2$ and it locks the grandparent block of the block $B_{v+2}$. Hence, $last\_locked\_view$ will be the view $v$. In order to vote, each replica makes sure that the received block satisfies the voting conditions. Voting conditions include $1)$ the proposed block extends the block at the $last\_locked\_view$ or $2)$ the view number of the received block's parent is greater than the $last\_locked\_view$. If either of the voting conditions is satisfied, then the replica will vote and update its $last\_voted\_view$ and the $last\_locked\_view$ to some new values. 
For example, when the replica receives the block for view $v+3$ ($B_{v+3}$), it then makes sure the voting conditions are met. As it can be seen, the block $B_{v+3}$ extends the  $last\_locked\_view$ (view $v$), satisfying the first condition (though the second condition is also satisfied as the view of the parent of $B_{v+3}$ is greater than $v$ or $v+2>v$). Therefore, the replica will vote for $B_{v+3}$. The replica then increments its $last\_voted\_view$ to $v+3$, its $last\_locked\_view$ to $v+1$. Then the replica checks its \textit{block tree} to see if there is any block that needs to be committed. If there are three blocks added over the top of each other during consecutive (uninterrupted) views ($B_v, B_{v+1}, B_{v+2}$) and extended by at least one another block (in this case $B_{v+3}$), then the replica will commit the block $B_v$ and all its predecessor blocks. The first three blocks added to the chain during consecutive uninterrupted views make a two-direct chain. For the first block in the two-direct chain to get committed, an additional block (may not form a direct chain) is need to extend it and make it a  three-chain.

Overall, the HotStuff primaries propose blocks, and replicas vote on the proposed blocks. Replicas commit a block if there is a  two-direct chain extended by another block, making a three-chain. Replicas move to the next view either when they receive a valid block with a $QC$ or when the primary fails to propose the block, and replicas timeout (wait for a proposal until timeout period). 

\begin{figure}
    \centering
    \includegraphics[width=8cm,height=1.2cm]{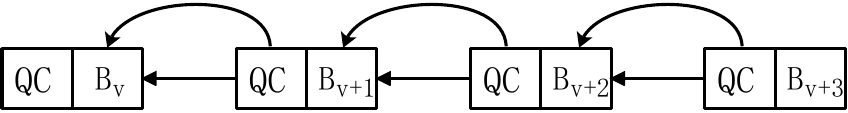}
    \caption{\textbf{The chain structure of pipelined HotStuff.}  Curved arrows denote the Quorum certificate references.}
    \label{fig:Pipelined-HotStuff}
\end{figure}

\section{Limitations of HotStuff}
\label{Section:Tradeoffs}
\subsection{Higher Latency During Happy Path}
 To achieve linear view change and responsiveness, HotStuff adds a round of communication. \textcolor{black}{
 In HotStuff the message propagation depends on the primary. Therefore, at the end of the second round, a replica cannot commit a block because it is not sure if other replicas will eventually complete the same round successfully. Only at the end of the three successful consecutive rounds for a block, a replica can be sure that other replicas at least have completed the first two rounds successfully. Though at the end of the second round, a replica in HotStuff locks on the $last\_locked\_view$, as it suspects some replicas might have committed the value at the $last\_locked\_view$. The lock is only released if the primary proves that no one has committed the value at the $last\_locked\_view$.}

 \textcolor{black}{
 The lock provides the safety guarantee, but a third round is required to achieve liveness. If a replica commits a block after two rounds, then the $last\_locked\_view$ will be the view of the parent of the most recent block. In this scenario, there can be a case where the $QC$ lock generated for the $last\_locked\_view$ is missing for the majority of nodes infinitely, hence creating a hidden lock problem} \cite{No-Commit-Proofs, Hot-stuff}.  
 
 \textcolor{black}{
 For example, during a view $v-1$ a primary proposes a block $b$. Upon receipt of $b$, each replica votes for $b$ and the primary for the view $v$ collects votes and builds a $QC$ ($qc$) certificate/lock for $b$. The primary of $v$ can only send a new proposal $b'$ (containing the $QC$ for $b$)  to a single replica $r_v$ before it fails. Now replica $r_v$ has the most recent lock (the lock/$QC$ for the view $v$). Replicas move to the view $v+1$ and the primary of the view $v+1$ receives responses (containing the lock for $v-1$ as the most recent lock) from $2f+1$ replicas (not including $r_v$). The primary for the view $v+1$ prepares a message $b'$ and adds the $QC$ $qc'$ (such that $qc'.view=v-1$). 
 Since $qc'$ is not the latest lock and $b'$ conflicts with $b$ (forming a fork), replica $r_v$ rejects it and remains locked with the lock $qc$. The primary for the view $v+1$ sends $b'$ only to one replica $r_{v+1}$, before it fails. The replica $r_{v+1}$ gets locked with the $QC$ for block $b'$ in the view $v+1$. This process can repeat indefinitely. }
 
\textcolor{black}{
Replicas $r_v$ and $r_{v+1}$ reject the proposal as they are holding a more recent lock ($QC$) than the one in the proposal. Unlocking these replicas requires assuring them that the block (they have locked) has not been committed. 
 In Fast-HotStuff, during a primary failure, the next primary aggregates and sends back all the locks it has received from replicas as proof (while only two locks need to be verified). This allows any replica holding the higher lock to release it (as it knows that the locked value has not been committed by any replica) and take part in consensus.}

 Therefore, the HotStuff protocol achieves consensus during the happy path in three rounds (whereas many other BFT protocols often use two rounds) \cite{Castro:1999:PBF:296806.296824,SBFT, jalalzai2020hermes, BFT-SMART}. Such additional 
 latency might discourage blockchain developers from building applications on top of HotStuff, as users have to wait more time in order to receive the notification that their transactions have been committed. 

\subsection{Forking Attack}
A Byzantine primary in pipelined HotStuff can deliberately generate forks to override the blocks from the correct (honest) primary~\cite{niu2021}. As it can be seen in Figure \ref{fig:Fork-attack}, blocks $B_{v+1}$ and  $B_{v+2}$ are honest blocks and the Byzantine primary will propose the block $B_{v+3}$ using the $QC$ for the block $B_v$. Other replicas will vote for $B_{v+3}$ as it satisfies the voting rule (
$B_{v+3}$ extends $B_{v}$  ). Since $B_{v+3}$ has higher view than $B_{v+2}$, next primary will extend $B_{v+3}$. Byzantine primary may choose to use the $QC$ from the block $B_{v+1}$ to override only a single block $B_{v+2}$. This may happen in case when the block $B_v$ has been proposed by another Byzantine primary. 

As a result of forking, the resources (computing, bandwidth, etc.) spent on forked blocks are wasted.  Similarly,
transactions of reverted blocks will have to be re-proposed by another primary. 
Hence, these transactions encounter additional latency.  Furthermore, since forking overrides the honest blocks it also breaks the direct chain relation among blocks (refer to Section \ref{Section: System Model}). 
This results in delaying the block commitment as the commit rule says that for a block to get commit, at least the first two blocks over the top of it should make a two-direct chain (two blocks from consecutive views). 
Forking attack makes it difficult to design an incentive mechanism like PoS over the top of HotStuff as blocks from correct primaries will not be able to get committed, and hence correct primaries will be less likely to earn rewards\footnote{More details about this will be provided in our future work.}.
  \begin{figure}
    \centering
    \includegraphics[width=6.5cm,height=2cm]{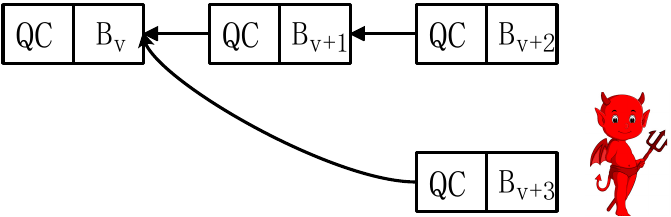}
    \caption{\textbf{Forking attack by the primary of the view v+3.}}
    \label{fig:Fork-attack}
\end{figure}

 \begin{figure}
    \centering
    \includegraphics[scale=0.33]{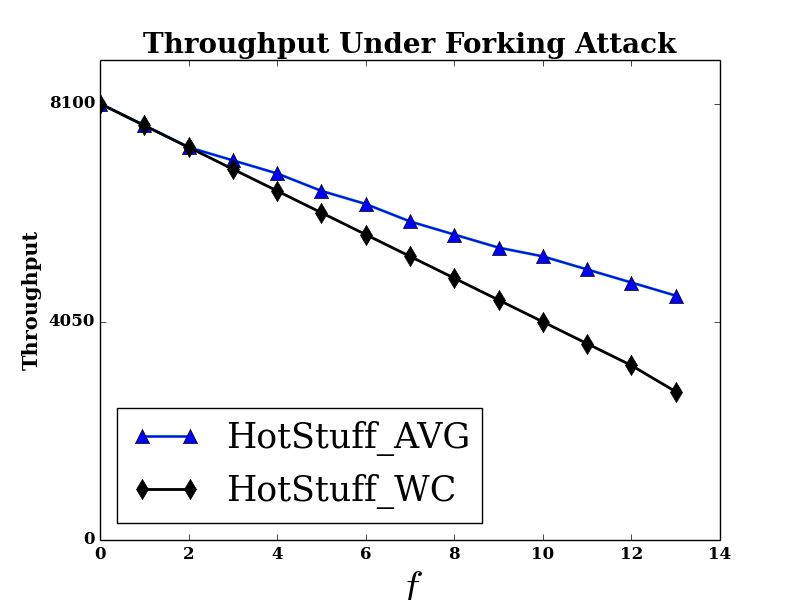}
    \caption{\textbf{Throughput decreases under Forking Attack.}
    \label{fig:Load-Tolerance}}
\end{figure}

   \begin{figure*}
    \centering
    \includegraphics[width=16cm,height=3cm]{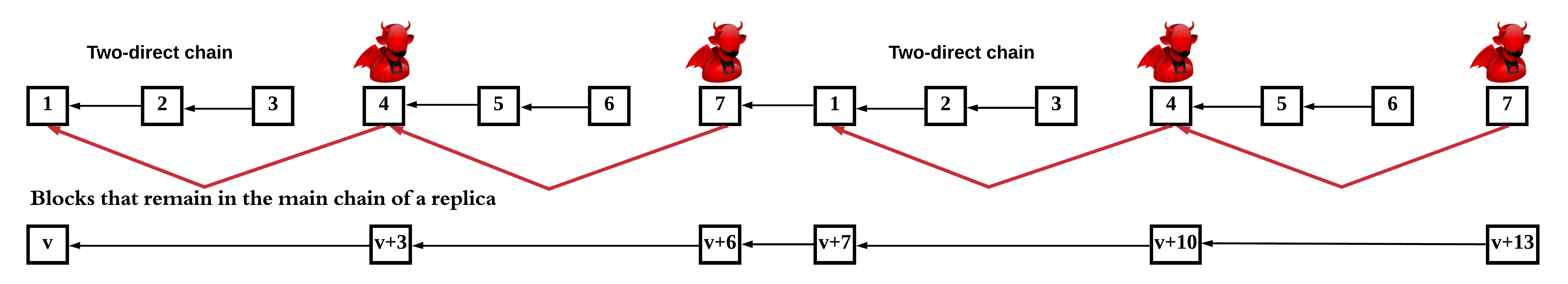}
    \caption{\textbf{ Worst case forking attack that is possible under round-robin mechanism of primary selection.}}
    \label{fig:Liveness-Failure-New}
\end{figure*}
  \textbf{Average Performance Due to Forking.}
    It is  important to know how HotStuff's average performance is affected by the forking attack. To find this out, we performed an experiment to show the tension between the number of faulty (Byzantine) replicas and the throughput in HotStuff as shown in Figure \ref{fig:Liveness-Failure-New}.   The network size $n=40$ and $f=13$.   The fraction of Byzantine replicas was initially zero and gradually increases to $f$. The throughput achieved in the absence of forking attack (in the absence of Byzantine replicas) is $8100$ tx/sec, but as the fraction of $f$ increases ($x-axis$) the average capacity of the network to process transactions decreases under forking attack as shown by the blue curve. 

\textbf{Worst Case Performance Due to Forking.}
As previously mentioned, each Byzantine primary can override at most two blocks from correct primaries. Worst case is possible when primaries are selected in  round-robin manner.
We show a pattern of primary selection, where the HotStuff will fulfill the two-direct chain commit condition only once. This means, after $n$ views, only once a  block (and any prior uncommitted block in its chain) that meet the commit condition will be committed.
In this case, each Byzantine primary avert two blocks proposed by correct/honest primaries. Hence, $f$ Byzantine primaries will avert $2f$ honest blocks. As a result, during $n$ views (assuming no timeout occurs), only $f+1$ blocks get committed. In these $f+1$ blocks committed during $n$ views only $1$ block is from a correct primary. Whereas, the remaining $f$ committed blocks are from Byzantine primaries. This empowers Byzantine primaries to censor or significantly delay specific transactions.
As it can be seen in Figure \ref{fig:Liveness-Failure-New} two-direct chain condition is satisfied after $n$ views (during $v+2$ and $v+9$, when primary $3$ is in charge). Once two-direct condition is met, for example in view $v+9$, $f+1=2+1=3$ blocks from views $v+7, v+6$ and $v+3$ get committed. Whereas primaries in views $v+7$ and $v+3$ are Byzantine. 

The worst-case throughput (under the different number of Byzantine replicas), in HotStuff, is presented by the black curve in the Figure \ref{fig:Load-Tolerance}.
As mentioned, only $f+1=13+1$ blocks from primaries will end up in the main chain. Where $f$ number of blocks belonging to Byzantine primaries and only one block is from a correct primary. Here, Byzantine primaries  do not reduce the size of their proposed block. The throughput of the network can be further reduced if Byzantine primaries propose smaller blocks. \textcolor{black}{It should be noted that the worst case can be avoided by determining the primary rotation
order via opening a common coin \cite{Common-Coin1,Random-Oracles-in-Constantinople}. }


 \section{ Fast-HotStuff}
\label{Section: Design of Fast-HotStuff}
Fast-HotStuff operates in a series of views with monotonically increasing view numbers. Each view number is mapped to a unique dedicated primary known to all replicas. 
The basic idea of Fast-HotStuff is simple. The primary has to convince voting replicas that its proposed block is extending the block pointed by the $highQC$. 
Once a replica is convinced,
it checks if a two-chain is formed over the top of the parent of the block pointed by the $highQC$ (the first chain in the two-chain formed has to be a one-direct chain in case of pipelined Fast-HotStuff). 
Then a replica can safely commit the parent of the block pointed by the $highQC$.

Unlike HotStuff (where a block is committed in three rounds), in \fangyu{the} Fast-HotStuff protocol, a replica $i$ can optimistically commit/execute a block during 
at the end of \fangyu{the} second round while guaranteeing that all other correct replicas will \fangyu{also} commit the same block in the same sequence eventually. 
This guarantee is valid when either one of the two conditions is met: 
\textbf{1) the block proposed by the primary is built by using the latest $QC$ that is held by the majority of replicas, or 2) by a $QC$ higher than the latest $QC$ being held by a majority of  replicas.}
Therefore, the primary has to incorporate proof\footnote{The proof  of $highQC$ is the proof provided by the primary in the proposed block that shows the latest $QC$ in the block is the latest $QC$ held by the majority of replicas or the latest $QC$  included in the block is more recent than the latest $QC$ held by the majority of replicas.} of the $highQC$ (being held by the majority of replicas or higher than the $highQC$ held by the majority of replicas) in every block it proposes, which can be verified by every replica.

Interestingly, the inclusion of the proof of the $highQC$ by the primary in Fast-HotStuff \fangyu{also} enables us to design a responsive two-chain consensus protocol. Indeed, the presence of the proof for the $highQC$ guarantees that a replica can safely commit a block after two-chain without waiting for \fangyu{the} maximum network delay as done in two-chain HotStuff \cite{Hot-stuff}, Tendermint \cite{tendermint}\chen{,} and Casper \cite{Casper}. Therefore, Fast-HotStuff achieves responsiveness only with a two-chain structure (which means two rounds of communication) in comparison to the three-chain structure in HotStuff (three rounds of communication).

     \SetKwFor{Upon}{upon}{do}{end}
\SetKwFor{Check}{check always for}
{then}{end}
\SetKwFor{Checkuntil}{check always until}
{then}{end}
\SetKwFor{Continuteuntil}{Continue until}
{then}{end}

\begin{algorithm}
\DontPrintSemicolon 
\small
\SetKwBlock{Begin}{Begin}{}
\SetAlgoLined
\caption{Utilities for replica \textit{i}}
\label{Algorithm: Utilities for replica i}
\SetAlgoLined

\SetKwFunction{Function}{SendMsgDown}
\SetKwFunction{FuncGetMsg}{MSG}
\SetKwFunction{VoteForMsg}{VoteMSG}
\SetKwFunction{CreatePrepareMsg}{CreatePrepareMsg}
\SetKwFunction{BQC}{GenerateQC}
\SetKwFunction{PQC}{ThresholdGenerateQC}
\SetKwFunction{Matchingmsgs}{matchingMsg}
\SetKwFunction{CreateAggQC}{CreateAggQC}
\SetKwFunction{MatchingQCs}{matchingQC}
\SetKwFunction{safeblock}{BasicSafeProposal}
\SetKwFunction{PipelinedSafeBlock}{PipelinedSafeBlock}
\SetKwFunction{aggregateMsg}{AggregateMsg}
\SetKwProg{Fn}{Func}{:}{}

\Fn{\CreatePrepareMsg{type, aggQC, qc, cmd}}{
b.type $\gets$ type\\
b.AggQC $\gets$ aggQC\\
b.QC $\gets$ qc\\
b.cmd $\gets$ cmd\\
return \textcolor{blue} {{b}}
}
\textbf{End Function}\\

\Fn{\BQC{$V$}}{
qc.type $\gets$ m.type : $m \in V$\\
qc.viewNumber $\gets$ m.viewNumber : $m \in V$\\
qc.block $\gets$ m.block :  $m \in V$\\
qc.sig $\gets$ AggSign( qc.type, qc.viewNumber, qc.block,i, \{m.Sig $|  m \in V$\})\\

return \textcolor{blue} {qc}
}
\textbf{End Function}\\


\Fn{\CreateAggQC{$\eta_{Set}$}}{
aggQC.QCset $\gets$ extract QCs from $\eta_{Set}$\\
aggQC.sig $\gets$
AggSign( curView, \{$qc.block |qc.block \in aggQC.QCset $\}, \{$i|i \in N$\}, \{$m.Sig |  m \in \eta_{Set}$\})\\
return \textcolor{blue} {aggQC}
}
\textbf{End Function}\\

\Fn{\safeblock{b, qc}}{
return \textcolor{blue}{\textit{b extends from $qc.block$}
}}
\textbf{End Function}\\

\Fn{\PipelinedSafeBlock{b,qc, aggQC}}{
\If{QC}{
return \textcolor{blue}{$b.viewNumber \geq curView \land b.viewNumber==qc.viewNumber+1$}
}
\If{AggQC}{
$highQC \gets$ extract highQc from $AggregatedQC$ \\
return \textcolor{blue}{\textit{ b extends from $highQC.block$}}
}

}
\textbf{End Function}
\end{algorithm}

Since HotStuff has two three-chain variants (the basic and the pipelined HotStuff), therefore we present two-chain basic Fast-HotStuff as well as pipelined Fast-HotStuff protocols. First, we present the basic optimized two-chain Fast-HotStuff and then extend it to an optimized two-chain pipelined Fast-HotStuff protocol.
Moreover, basic Fast-HotStuff needs the proof for the $highQC$ in the form of $n-f$ $QC$s attached in the proposed block by the primary in order to be able to guarantee safety and liveness as a two-chain protocol. In the case of blockchains since block size is usually large (multiples of Megabytes), this overhead has little effect on performance metrics like throughput and latency. We also show that the protocol requires only two $QC$s to be verified instead of $n-f$ $QC$s.
Then, we further optimized pipelined Fast-HotStuff by introducing proposal pipelining. Moreover, for pipelined Fast-HotStuff, the block includes a single $QC$ during the happy path, and similar to the basic Fast-HotStuff only two $QC$s need to be verified in case of the primary failure. 

 \subsection{Basic Fast-HotStuff}
 The algorithm for two-chain responsive basic HotStuff is given in Algorithm \ref{Algorithm: Basic-Fast-HotStuff}. 
 Below, we describe how the basic Fast-HotStuff algorithm operates. Fast-HotStuff operates in three phases \textsc{prepare}, \textsc{precommit} and \textsc{Commit}. The function of each phase is described below.

 \textbf{PREPARE phase.}
  Initially, the primary replica waits for  new-view messages $\eta = \langle "NEWVIEW" , curView, prepareQC,i \rangle$ from $n-f$ replicas. The \textsc{NEWVIEW} message contains four fields indicating message type (\textsc{NEWVIEW}),  current view ($curView$), the latest $PrepareQC$  
  known to replica $i$, and replica id $i$.  The $NEWVIEW$ message is signed by each replica over  fields $\langle curView, prepareQC.block, i \rangle$. $prepareQC.block$ is presented by the hash of the block for which $prepareQC$ was built (hash is used to identify the block instead of using the actual block). 
  The primary creates and signs a prepare message ($B=\langle "Prepare", AggQC, commands, curView, h,i \rangle$)  and broadcasts it to all replicas, as shown in Algorithm \ref{Algorithm: Basic-Fast-HotStuff}. 
  We can also use the term block proposal or simply block for $PREPARE$ message. 
  $AggQC$ is the aggregated $QC$ build from valid $\eta$ messages collected from $n-f$ replicas.

  Upon receipt of a $PREPARE$ message $B$ from the primary, a replica $i$ verifies $AggQC$,  extracts \textcolor{black}{the $QC$ with the latest/highest view among $QC$s} ($highQC$) as well from $PREPARE$ message, and then checks if the proposal is safe. Verification of $AggQC$ involves verification of aggregated signatures built from $n-f$ $\eta$ messages and verification of the latest $PrepareQC$. Signature verification of each $QC$ is not necessary, a replica only needs to make sure messages are valid. 
  $BasicSafeProposal$ predicate \textcolor{black}{(shown in Algorithm \ref{Algorithm: Utilities for replica i} and used in Algorithm \ref{Algorithm: Basic-Fast-HotStuff}, line 12)} makes sure a replica only accepts a proposal that extends from the $highQC.block$. \textcolor{black}{In other words, a proposal is safe if the primary provides the proof of $highQC$ for the $QC$ in the proposal.}
  
  If a replica notices that it is missing a block, it can download it from other replicas.
  At the end of \textsc{PREPARE} phase, each replica sends its vote $v$  to the primary. Any vote message $v$ sent by a replica to the primary is signed on tuples $\langle type, viewNumber, block, i\rangle$. Each vote has a type, such that $v.type \in \{PREPARE, PRECOMMIT\}$ (Here again block hash can be used to represent the block to save space and bandwidth).
  
  \textbf{PRECOMMIT phase.}
 The primary collects  \textsc{PREPARE} votes from $n-f$ replicas and builds  $PrepareQC$. The primary then broadcasts $PrepareQC$ to all replicas in \textsc{PRECOMMIT} message. Upon receipt of a valid \textsc{PRECOMMIT} message, a replica will respond with a \textsc{PRECOMMIT} vote to the primary. Here (in line 25-27 Algorithm \ref{Algorithm: Basic-Fast-HotStuff}) replica $i$ also checks if it has committed the block for $highQC$ called $highQC.block$. Since the majority of replicas have voted for $highQC.block$, it is safe to commit it. \looseness=-1
 
   \textbf{COMMIT phase.}
   Similar to \textsc{PRECOMMIT} phase, the primary collects  \textsc{PRECOMMIT} votes from $n-f$ replicas and combines them into $PrecommitQC$. As a replica receives and verifies the $PrecommitQC$, it executes the commands. The replica increments $newNumber$ and begins the next view.
   
    \textbf{New-View.}
    Since a replica always receives a message from the primary at a specific viewNumber, therefore, it has to wait for a timeout period during all phases. \textcolor{black}{NEXTVIEW is an auxiliary utility that determines the timeout period. As a replica waits for the receipt of a message (during any phase),} if NEXTVIEW (viewNumber) utility interrupts waiting, the replica increments viewNumber and starts the next view.

\begin{algorithm}
\small
\SetKwFor{ForEachin}{foreachin}{}{}
\caption{Basic Fast-HotStuff for replica $i$}
\label{Algorithm: Basic-Fast-HotStuff}
\ForEachin{curView $\gets$ 1,2,3,... }{ 
\textcolor{blue}{$\triangleright$  Prepare \text{ Phase}} \\
\If{{i is primary}}{
\text{wait until $(n - f)$  $\eta$ messages are received}: 
    $\eta_{Set} \gets \eta_{Set} \cup \eta $ \\
    $aggQC \gets CreateAggQC(\eta_{Set}) $\\
     $B\gets$ CreatePrepareMsg(\textit{Prepare,$aggQC$, client’s command})\\

    broadcast $B$

}
\If{{i is normal replica}}{
 
 \text{wait for prepare $B$  from primary(curView)}\\
$highQC \gets$ extract highQc from $B.AggQC$\\


\If{BasicSafeProposal($B$, $highQC$)}  {

Send vote $v$ for prepare message to primary(curView)

}

}
\textcolor{blue}{$\triangleright$  $Pre-Commit$  \text{ Phase}} \\
\If{{i is primary}} {
\text{wait for ($n - f$) prepare votes}: V 
$\gets V \cup v$ \\

PrepareQC $\gets$ BasicGenerateQC(V)\\
broadcast PrepareQC
}

\If{{i is normal replica}}{
\text{wait for  $PrepareQC$  from primary(curView)}\\
Send Precommit vote $v$ to primary(curView)\\
\If{have not committed $highQC$.block}{
commit $highQC$.block
}

}
\textcolor{blue}{$\triangleright$  $Commit$ Phase} \\
\If{{i is primary}}{

wait for $(n - f)$ votes: $V \gets V \cup v$ \\   
PrecommitQC $\gets$ GenerateQC(V)\\
broadcast PrecommitQC\\
}
\If{{i is normal replica}}{
\text{wait for  PrecommitQC  from primary(curView)}\\

execute new commands through $PrecommitQC$.block\\ respond to clients\\
\textcolor{blue}{$\triangleright$  New-View} \\
\Check{
nextView interrupt} {goto this line if nextView(curView) is called during “wait for” in any phase\\
Send $\eta$ to
primary($curView + 1$)\\
}
}
}

\end{algorithm}

\textbf{Efficient View Change\textcolor{black}{.}} 
As a replica moves to the next view, it will send its $NEWVIEW$ ($\eta$) to the next primary. The primary aggregates $n-f$ $\eta$ messages and their signatures into an $AggQC$ and sends back to all replicas. 
Upon receipt of a block containing an $AggQC$, first the aggregated signature for the $n-f$  ($\eta$) messages needs to be verified.
The first check (verification of aggregated signature of $AggQC$) verify that the  $AggQC$ that is built from $\eta$ messages has $n-f$ valid $QC$s ($QC$s come from $n-f$ distinct replicas). It further guarantees that at least $f+1$ $QC$s out of $n-f$ are from correct replicas. 

Therefore, a replica only needs to find a $QC$ with the highest view among $QC$s in $AggQC$ by looping over the view numbers of $QC$s. Next, the replica has to  verify the aggregated signature of $highQC$ or latest $QC$. 
As a result, we do not need to verify the remaining $n-f-1$ $QC$s as a replica only needs to verify  $highQC$ and make sure if the block extends $highQC.block$. 
This helps to reduce the signature verification cost for $AggQC$ ($n-f$ $QC$s) during view change by the factor of $O(n)$ independent of signature scheme used. 
This optimization can be used by any two-chain BFT-based consensus protocol during view change where replicas have to receive and process $n-f$ $QC$s.


    \textcolor{black} {This optimization can also be used by the primary while receiving $\eta$ messages. The primary only needs to verify the signature of the sender of the $\eta$ message and the latest $QC$ ($highQC$) among $n-f$ $QC$s ($prepareQC$ in $\eta$).} 

 \textcolor{black}{While receiving a block with the $AggQC$ from the primary,
there is a possibility that the $\eta$ message containing $highQC$ is invalid (because the primary is malicious).    It means $\eta$ does not meet formatting requirements, its view number, or the $highQC$ is invalid (its aggregated signature cannot be verified). Formatting requirements involve incorrect message types or fields.  In this case, the replica can reject the block proposal. Furthermore, a replica can use this proof to punish (blacklist) the primary. The proof can be shared with other replicas and a correct primary will use the proof to propose a transaction to blacklist the Byzantine primary \cite{jalalzai2020hermes}.}
\subsection{Pipelined Fast-HotStuff}

 Pipelined Fast-HotStuff has been optimized in different ways in comparison to the basic Fast-HotStuff. First, similar to the pipelined HotStuff, it pipelines requests and proposes them in each phase to increase the throughput. Secondly, during normal view change when no primary failure occurs, the protocol only requires the block to carry a single $QC$ instead of $n-f$ $QC$s. The proof of  $highQC$ in pipelined Fast-HotStuff carries a small overhead ($AggQC$) in the block during view $v$ if the primary in view $v-1$ fails.  The pipelined Fast-HotStuff algorithm can be summarized as follows: The primary of the view $v$ collects $n-f$ votes, builds a $QC$, adds the $QC$ to the block, and propose the block for the view $v$. If the view $v-1$, has failed, then the primary, collects $n-f$ $\eta$ (\textsc{NEWVIEW}) messages. Subsequently, the primary aggregates $\eta$ messages into $AggQC$, adds $AggQC$ to the block and proposes the block for the view $v$.

 Blocks can be added into the chain either during the happy path with no failure or the primary fails, and the next primary will have to add its block into the chain.
Unlike HotStuff, pipelined Fast-HotStuff addresses these two cases differently. Indeed, there are two ways a primary can convince replicas that the proposed block extends the $highQC$. 
In \textcolor{black}{the} case of contiguous chain growth (happy path), a primary can only propose a block during the view $v$ if it can  build a $QC$ from $n-f$ votes received during the view $v-1$. 
Therefore, in a contiguous case, each replica signs the vote 
, the primary will build a $QC$ from $n-f$ received votes and include it in the next block proposal. Therefore, the block will contain only the $QC$ for view $v-1$. As a result, for the block during view $v$, the $QC$ generated from votes in $v-1$ is the proof of $highQC$. It should be noted that in pipelined Fast-HotStuff, a replica commits a block if two-chain (with the first chain needing to be a direct-chain) is formed over the top of it. 
Similarly, if a primary during view $v$ did not receive $n-f$ votes from view $v-1$ (this means the primary during view $v-1$ has failed), then it can only propose a block if it has received $n-f$ $\eta$ (\textsc{NEWVIEW}) messages  from distinct replicas for view $v$. In this case, the primary during view $v$ has to propose a block with aggregated $QC$ or $AggQC$ from $n-f$ replicas.  

The signatures on $\eta$ are aggregated to generate a single aggregated signature in $AggQC$.
Therefore, two types of blocks can be proposed by a primary: a block with $QC$ if the primary is able to build a $QC$ from \textcolor{black}{the} previous view 
or a block with $AggQC$ 
if the previous primary has failed.
It should be noted that to verify $AggQC$ a replica only needs to verify the aggregated signature of $AggQC$ and the $highQC$ in the $AggQC$ as described previously. 
The $highQC$ can simply be found by looping over view numbers of each $QC$ and \textcolor{black}{choosing} the highest one.

The chain structure in pipelined Fast-HotStuff is shown in Figure \ref{fig:Chained-Fast-HotStuff}. Here, we can see that upon receipt of  block $B_{v+2}$ (during the view $v+2$), a two-chain with a direct chain is completed for the block $B_v$. Now a replica can simply execute  the block $B_v$. Since there is no primary failure for  views $v$ through view $v+2$, only $highQC$ is included in the block proposal. Similarly, the primary for the view $v+3$ has failed, therefore, the primary for the view $v+4$, has to propose a block with $AggQC$ (small overhead). Upon receipt of $B_{v+4}$, the $highQC$ can be extracted from $AggQC$. In this case, the $QC$ for the block $B_{v+2}$
has been selected as the $highQC$.

 As it can be seen the algorithm for the pipelined Fast-HotStuff mainly has two components: the part executed by the primary and the part run by replicas. The \textit{primary} either receives votes that it will aggregate into a $QC$ or $\textsc{NEWVIEW}$ messages containing $QC$ from $n-f$ replicas   that the primary will aggregate into $AggQC$ (Algorithm \ref{Algorithm: Chained-Fast-HotStuff} lines 23-25 and 2-9). The primary then builds a proposal in the form of \textsc{PREPARE} message also called a block. \textsc{PREPARE} also contains $QC$ or $AggQC$ depending on if it has received $n-f$ votes or \textsc{NEWVIEW} messages (Algorithm \ref{Algorithm: Chained-Fast-HotStuff} lines 2-11). The primary then proposes the block to replicas ((Algorithm \ref{Algorithm: Chained-Fast-HotStuff} line 10)).
 
 Upon receipt of block $B$ (containing a $QC$) through a proposal, each replica can check the condition $B.viewNumber==B.QC.viewNumber+1$ and accept the proposal if the condition is met. On the other hand, if the block contains a valid $AggQC$, then each replica
   extracts the  $highQC$ from $AggQC$ and checks if the block extends the $highQC.block$. These checks are performed through $PipelinedSafeBlock$ predicate (Algorithm \ref{Algorithm: Utilities for replica i}).
   If the check was successful, each replica sends back a vote to the next primary (Algorithm \ref{Algorithm: Chained-Fast-HotStuff} lines 13-16).
\looseness=-1
   After that, each replica commits the grandparent of the received block if a direct-chain is formed between the received block's parent and its grandparent.
   (Algorithm \ref{Algorithm: Chained-Fast-HotStuff}, lines 17-21). 
   Hence, once the commit condition is met for a block in pipelined Fast-HotStuff, the block is committed after two rounds of communication
   without waiting for the maximum network delay.
   
       \textbf{Resilience against Forking Attack }. Unlike pipelined HotStuff, pipelined Fast-HotStuff is robust to forking  attacks. In Fast-HotStuff the primary has to provide the proof of the latest $QC$ (known to the majority of replicas or a $QC$ more recent than the latest one known by the majority of replicas) included in the block. This proof can be provided in two ways. First, if there is no primary failure, then for the proposed block $B^*$ and the $QC$ it contains ($qc$), we have $B.view=qc.view+1$. Secondly, if there is a primary failure in the previous view then the primary has to include the $AggQC$ ($n-f$ $QC$s). This guarantees at least one of the $QC$s in the $AggQC$ is the latest $QC$ held by the majority of replicas or a $QC$ higher than the latest $QC$ being held by the majority of replicas. The inclusion of proof within \textcolor{black}{the} block prevents the primary from using an old $QC$ to generate forks. If a primary does not provide appropriate proof, its proposal can be rejected. Furthermore, such a proposal can be used as proof to blacklist the primary (by proposing a blacklisting transaction with the proof).

\looseness=-1

 \begin{figure}
    \centering
    \includegraphics[width=9cm,height=2.5cm]{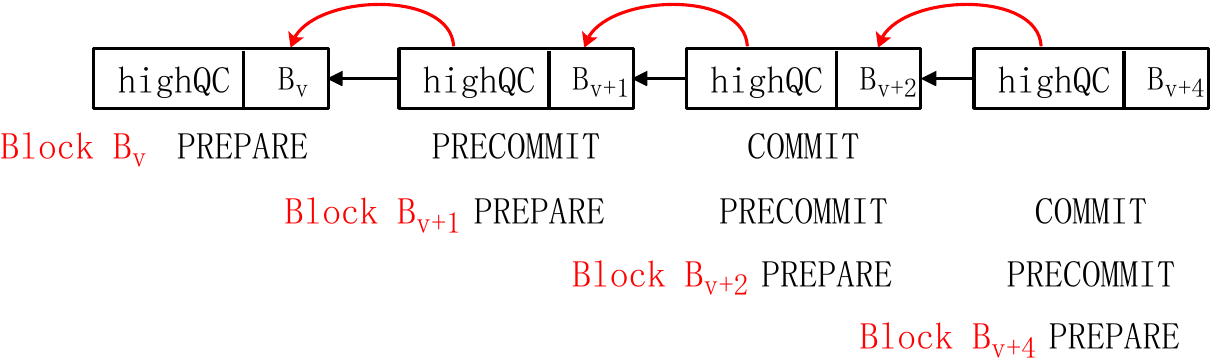}
    \caption{\textbf{Pipeined/Chained Fast-HotStuff where a QC can serve in different phases simultaneously. Note that the primary for view $v+3$ has failed.}}
    \label{fig:Chained-Fast-HotStuff}
\end{figure}

 \begin{algorithm}[ht]
\SetKwFunction{Function}{SendMsgDown}
\SetKwFunction{FuncGetMsg}{MSG}
\SetKwFunction{VoteForMsg}{VoteMSG}
\SetKwFunction{CreatePrepareMsg}{CreatePrepareMsg}
\SetKwFunction{QC}{GenerateQC}
\SetKwFunction{Matchingmsgs}{matchingMsg}
\SetKwFunction{MatchingQCs}{matchingQC}
\SetKwFunction{safeblock}{SafeProposal}
\SetKwFunction{aggregateMsg}{AggregateMsg}
\SetKwProg{Fn}{Func}{:}{}

\small

\SetKwFor{ForEachin}{foreachin}{}{}
\caption{pipelined Fast-HotStuff for block $i$}
\label{Algorithm: Chained-Fast-HotStuff}

\ForEachin{curView $\gets$ 1,2,3,... }{ 
\If{i is primary}{
\If{$n-f$ $\eta$ \textit{ msgs are received}}{
 $aggQC \gets    CreateAggQC(\eta_{Set})$

     $B\gets$ CreatePrepareMsg(\textit{Prepare,$aggQC$,$\perp$ , client’s command})\\
     }
    \If{$n-f$ $v$ \textit{ msgs are received}}{
       $qc \gets GenerateQC(V)$
     CreatePrepareMsg(\textit{Prepare,$\perp$, $qc$, client’s command})
    }
    broadcast $B$
    

}
\If{{i is normal replica}}{
 \text{wait for $B$  from primary(curView)}\\

\If{ PipelinedSafeBlock($B, B.qc, \perp$) $\lor$  PipelinedSafeBlock($B, \perp, B.aggQC $) }  {

Send vote $v$ for prepare message to primary(curView+1)

}

// start commit phase on $B^*$’s grandparent if  a direct chain exists between $B^*$'s parent and grandparent \\
 \If {$(B^{*}
.parent = B'')
\land (B''
.parent = B') \land B''.view=B'.view+1)
$} {
 execute new commands through $B'$\\
 respond to clients
}

}
\If{{i is next primary}}{
\text{wait until $(n - f)$  $v/\eta$ for current view are received}: $\eta_{Set} \gets \eta_{Set} \cup \eta \lor V \gets V \cup v$

}
\textcolor{blue}{$\triangleright$  Finally} \\
\Check{nextView interrupt}{
goto this line if nextView(curView) is called during “wait for” in any phase\\
Send $\eta$ to
primary($curView + 1$)\\
}
}
\end{algorithm}

\section{Proof of Correctness for HotStuff}
\label{Section: Proof-of-correctness}
\subsection{Correctness Proof for Basic Fast-HotStuff}

Correctness of a consensus protocol involves proof of safety and liveness. Safety and liveness are the two important properties of consensus algorithms. In this section, we provide proof \textcolor{black}{of} the safety and liveness properties of Fast-HotStuff. \textcolor{black}{  A proposal of view $v$ being committed by a replica $i$ means that the network has decided on the proposal. The node $i$ can execute the proposal to update its state, which reflects the ordered execution of the predecessor proposals.}
We begin with two standard definitions.
\begin{definition}[Safety]
A protocol is safe if the following statement holds: if at least one correct replica commits a block at the sequence (blockchain height) $s$ in the presence of $f$ Byzantine replicas,  then no other block will ever be committed at the sequence $s$. \label{Definition:Safety}
\end{definition}

\begin{definition}[Liveness]
A protocol is live if it guarantees progress 
in the presence of at most $f$ Byzantine replicas.
\label{Definition:Liveness}
\end{definition}
Next, we introduce several technical lemmas. The first lemma shows that for each view, at most one block can get certified (get $n-f$ votes).

\begin{lemma}\label{lemma: Different Seq}
If any two valid $QC$s, $qc_1$ and $qc_2$ with same type $qc_1.type=qc_2.type$ and conflicting blocks i.e., $qc_1.block=B$ conflicts with $qc_2.block=B'$, then we have $qc_1.viewNumber \neq qc_2.viewNumber$. 
\end{lemma}
\begin{proof}

We can prove this lemma by contradiction. Furthermore, we assume that $qc_1.viewNumber = qc_2.viewNumber$. Now let's consider, $N_1$ is a set of replicas that have voted for block $B$ in $qc_1$($|N_1| \geq 2f+1$). Similarly, $N_2$ is another set of replicas that have voted for block $B'$ and whose votes are included in $qc_2$ ($|N_2| \geq 2f+1$). Since $n = 3f+1$ and $f = \frac{n-1}{3}$,  this means there is at least one correct replica $j$ such that $j \in N_1 \cap N_2$ (which means $j$ has voted for both $qc_1$ and $qc_2$). But a correct replica only votes once for each phase in each view. Therefore, our assumption is false\textcolor{black}{,} and hence, $qc_1.viewNumber \not = qc_2.viewNumber$.
\end{proof}

The second lemma proves that if a single replica has committed a block at the view $v$ then all other replicas will commit the same block at the view $v$.

\begin{lemma}\label{lemma:PrepareQC-Block-extension}
If at least one correct replica $i$ has received $PrecommitQC$ for block $B$ (committed block $B$), then the  $PrepareQC$ for block $B$ will be the $highQC$ for next (child of block $B$) block $B'$.
\end{lemma}

\begin{proof}
  The primary begins with a new view $v+1$ as it receives $n-f$ \textsc{NEWVIEW} messages. Here for ease of understanding, we assume $v+1$.  $B'$ could have been proposed during any view $v'>v$.
We show that any combination of 
$n-f$ \textsc{NEWVIEW} messages  have at least  one of those \textsc{NEWVIEW} message received by the primary containing $highQC$ ($PrepareQC$ for block $B$) for block $B$.

We know that in \textcolor{black}{the} previous view $v$, replica $i$ has committed block $B$. This means a set of replicas $R_1$ of size  $|R_1|\geq 2f+1$ have voted for $PrepareQC$ ( which is built from votes for block $B$). Similarly, another set of  replicas $R_2$ of size $|R_2|\geq 2f+1$ have sent their \textsc{NEWVIEW} messages to the primary of view $v+1$ after the end of view $v$ . Since the total number of  replicas is $n=3f+1$ and $f = \frac{n-1}{3}$, therefore, $R_1 \cap R_2 =R $, such that, $|R| \geq f+1$, which means that there is at least one correct replica in $R_2$ that has sent its $PrepareQC$ as $highQC$ in \textsc{NEWVIEW} message to the primary. Therefore, when primary of view $v+1$ proposes $B'$, the $PrepareQC$ for block $B$ will be the $highQC$. In other words, $B'$ will point to $B$ as its parent (through $highQC$). 
\end{proof}
The third lemma proves that if a replica commits a block, then it will not be reverted. 
\begin{lemma} \label{lemma:safety}
A correct replica will not commit two conflicting  blocks.
\end{lemma}
\begin{proof}
From lemma \ref{lemma: Different Seq} we know that each correct replica votes only once for each view and therefore view number for conflicting blocks $B$ and $B'$ will not be the same. Therefore, we can assume that $B.curView < B'.curView$.
Based on lemma
\ref{lemma:PrepareQC-Block-extension} we know that if at least one correct replica has received  $PrecommitQC$ for block $B$   (have committed $B$), then the $PrepareQC$ for block $B$ will be the $highQC$ for the next block $B'$ (child of block $B$).  
Therefore, any combination of $n-f$ $PrepareQC$s in $AggQC$ for $B'$ will include at least one $highQC$ such that $highQC.block=B$ or $highQC.block.view=B.view$. But
for $B'$ to be conflicting with $B$ it has to point to the parent of $B$ at least. 
 Consequently, first\textcolor{black}{,} it is not possible for a primary to build a valid \textsc{PREPARE} message $B'$ in which $highQC.block.view < B.view$. Secondly, if a primary tries to propose a block with invalid $AggQC$, it will be rejected
 by the replica based on Algorithm \ref{Algorithm: Basic-Fast-HotStuff} line 12. 
\end{proof}

Lemmas \ref{lemma: Different Seq}, \ref{lemma:PrepareQC-Block-extension} and \ref{lemma:safety} provide a safety \textcolor{black}{(Definition \ref{Definition:Safety})} proof for basic Fast-HotStuff consensus protocol.
 To ensure liveness  \textcolor{black}{(Definition \ref{Definition:Liveness})}, Fast-HotStuff has to make sure in each view a replica is selected as a primary \textcolor{black}{ deterministically (eventually rotating through all replicas)} and the view number is incremented. Moreover, we should also show that the protocol will eventually add a block to the chain/tree of blocks (a block to the blockchain) or will result in view change in the case of failure.

\textcolor{black}{The NEXTVIEW maintains the timeout value (interval), and the timer is set upon entering a view}. If a replica times out (without reaching a decision), it can employ exponential-backoff used in PBFT \cite{Castro:1999:PBF:296806.296824} to double its timeout value and \textcolor{black}{calls the NEXTVIEW (to advance the view).} 
This guarantees that eventually after GST there will be a period of synchrony when timeout values from all correct replicas overlap \textcolor{black}{for a long enough interval ($T_f$),} the primary will be correct and at least $n-f$ honest/correct replicas will be in the same view. As a result, it can be shown that this  period is enough to reach a decision during consensus.
Below we provide liveness \textcolor{black}{(Definition \ref{Definition:Liveness})} proof for our Fast-HotStuff protocol. 
\begin{lemma}
\label{Lemma:Liveness correct primary}
After GST, there is a time $T_f$, when there is a correct primary and all correct/honest replicas are in the same view. As a result,  a decision is reached during  $T_f$.
\end{lemma}
\begin{proof}
 Based on Lemma \ref{lemma:PrepareQC-Block-extension}, at the beginning of a view, the primary will have the latest $PrepareQC$ or $highQC$ from $n-f$ replicas. As per assumption, \textcolor{black}{after GST the period that $T_f$ bound on message delivery holds is sufficiently large. Therefore,} when all correct replicas are in the same view (\textcolor{black}{due to exponential-backoff as explained}),  the correct primary will propose a \textsc{PREPARE} message with $AggQC$ containing the  latest $PrepareQC$, which is extracted by each replica. Since all replicas are in the same view until bounded $T_f$ time, therefore all replicas can successfully complete \textsc{PREPARE}, \textsc{PRECOMMIT}\chen{,} and \textsc{COMMIT} phase. 
\end{proof}

\subsection{Correctness Proof for Pipelined Fast-HotStuff}
We can establish the safety \textcolor{black}{(Definition \ref{Definition:Safety})} and liveness \textcolor{black}{(Definition \ref{Definition:Liveness})} of pipelined Fast-HotStuff in a way similar to what we have proven for the basic Fast-HotStuff. For safety, we want to prove that if a block is committed by a replica, then it will never be revoked. Moreover, if a replica commits a block then all other replicas will eventually commit the same block at the same block sequence/height.

 \begin{lemma} \label{Lemma:Pipelined-FHS}
 If B and B' are two conflicting blocks, then only one of them will be committed by a correct replica.
 \end{lemma}
 
 \begin{proof}
 From lemma \ref{lemma: Different Seq}, we know that each correct replica votes only once for each view and therefore view number for conflicting blocks $B.curView < B'.curView$ and we assume that block $B$ is already committed.
 Now,  a replica $r$ receives a \textsc{PREPARE} message in the form of block $B'$.  Since $B'$ is a conflicting block to $B$ therefore, the $highQC$ in the block must point to an ancestor of block $B$. In this case we consider $B^{*}$ as the parent of $B$ ($B.parent=B^{*}$) and $B'.QC.block=B^{*}$. 
 In case of $AggQC$,
 since $B'$ extends from $B^{*}$ which is parent of $B$, therefore first condition in $PipelinedSafeBlock$ predicate fails because $highQC$ in $B'$ ($B'.QC$) extends the parent of $B$ ($B^*$) but not $B$ (invalid $AggQC$). \textcolor{black}{Any valid $AggQC$ will have at least one $QC$ for block $B$ or any other block that extends from block $B$ (Lemma \ref{lemma:PrepareQC-Block-extension}).}
 Similarly,  second condition in $PipelinedSafeBlock$ predicate also fails because $B'.viewNumber=B'.QC.viewNumber+k$, where $k>1$. \looseness=-1
 \end{proof}

\begin{figure}
    \centering
    \includegraphics[width=9cm,height=2.5cm]{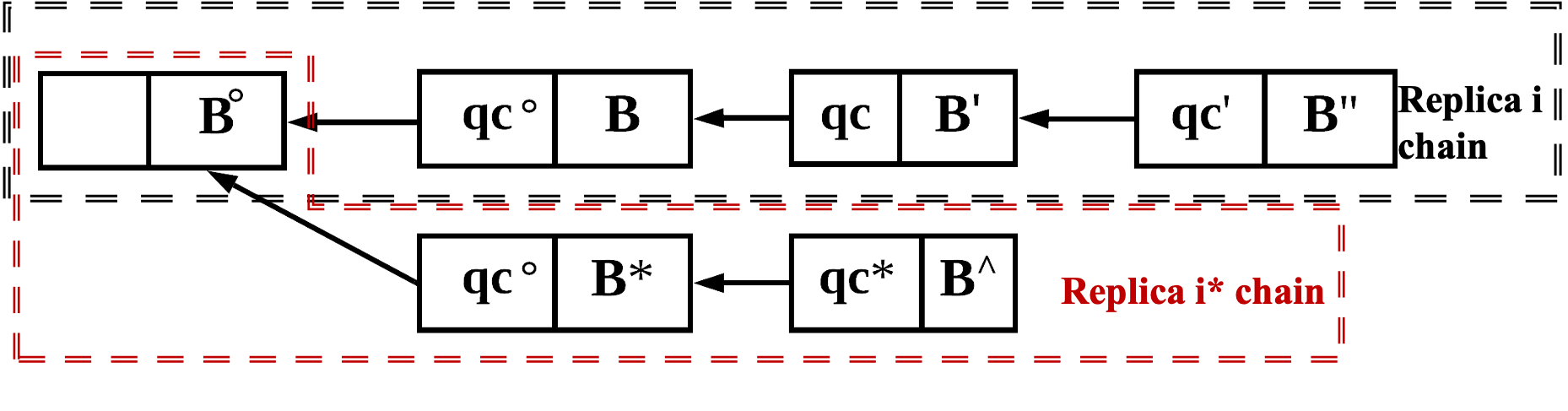}
    \caption{\textbf{Safety violation scenario}.}
    \label{fig:Safety-violation}
\end{figure}

\begin{lemma}
If a block B is committed by a \textcolor{black}{correct} replica $i$,  then no other \textcolor{black}{correct} replica $i^*$ in the network will commit another block $B^*$ \textcolor{black}{ that conflicts with the block $B$}.
 \end{lemma}

  \begin{proof}
 For the sake of contradiction let's assume that it is possible that if a block $B$ is committed only by a single correct replica $i$ then at least one another correct replica $i^*$ can commit \textcolor{black}{another block $B^*$ which conflicts with the block $B$.}
 The chain of each replica $i$ and $i^*$ is given in the Figure \ref{fig:Safety-violation}. Now we analyze different cases \textcolor{black}{that} arise from our assumption.

  Case 1: When $qc^*.view < qc.view$. In this case, since  $qc^*.view < qc.view$ therefore,  $qc^*.view < highQC.view$.  As a result, $qc^*$ will not be selected as $highQC$ by any correct replica  (based on Lemma \ref{lemma:PrepareQC-Block-extension}).
  
  Case 2: When $qc^*.view == qc.view$.
    Two $QC$s pointing to two different blocks cannot be the same.
 
 Case 3: When $qc'.view > qc*.view > qc.view$. In the third case,  there is a possibility that several replicas hold a $QC$ $qc^*$ such that $qc^*.view > qc.view$. $qc^*$ has not been propagated to the majority of correct replicas due to the network partition. 
  The $qc^*.block=B^*$ is conflicting block to block $B$. After the view change the new primary $p$ (that is aware of $qc^*$) includes $qc^*$ in the $AggQC$.  If $qc^*$ is accepted, then other replicas may prefer replica $i^*$'s chain (colored in red in Figure \ref{fig:Safety-violation}). As a result, replicas may
  commit the block $B^*$ which is conflicting to block $B$. But this is not possible  due to the two-direct chain requirements to commit a block  ($qc.view+1=qc'.view$) in Algorithm \ref{Algorithm: Chained-Fast-HotStuff} lines 18-20. Therefore, there can be no such $qc^*$, with a view $qc'.view > qc^*.view > qc.view$ when replica $i$ has committed the block $B$. 
  
  Case 4: When $qc^*.view == qc'.view$.
      Two $QC$s pointing to two different blocks cannot be the same.
      
 Case 5: When $qc^*.view > qc'.view$ 
 We know that $qc'$ has been built from $n-f$ replica votes that have seen the $qc$. Therefore, the $qc$ is the $highQC$ for $n-f$ replicas of which at least $f+1$ are correct. Hence, when the primary collect $n-f$ $\eta$ messages (\textsc{NEWVIEW}) during view change at least one of  $\eta$ will contain $qc$ or a $QC$ that extends the $qc$ (based on Lemma \ref{lemma:PrepareQC-Block-extension}). Therefore, $qc^*$ that does not extend $qc$ will not be formed.
  Hence, our assumption is proved to be false. Therefore, if a block is committed by \textcolor{black}{a correct replica $i$,  then no other conflicting block can be committed by any other correct replica.}

 \end{proof}
 
 This lemma confirms that if a single replica commits a block $B$, then no conflicting block to $B$ will be committed by another replica. This lemma is developed to address a subtle safety violation scenario due to the network partition that was reported in Gemini \cite{banogemini} (that simulates Byzantine scenarios at scale) by Facebook's Novi group. 
 We also successfully verified the safety of  Fast-HotStuff using the Twins simulator for Fast-HotStuff developed by \textcolor{black}{Novi}  team\footnote{https://github.com/asonnino/twins-simulator}.

\begin{lemma}
\label{Lemma:Pipelined liveness correct primary}
After GST, there is a time $T_f$, when there are at least two consecutive correct primaries and all correct/honest replicas are in the same view.  As a result,  a decision will eventually reach.
\end{lemma}
\begin{proof}
 As $T_f$ period is sufficiently large after the GST, while two correct primaries are selected consecutively and correct replicas are in the same view, a one-direct chain will be formed between blocks proposed by the correct primaries. Since two-chain (with the first chain as one-direct chain) are required for a block to get committed, the third block will eventually be added by another correct primary. 
 For brevity, more details about pipelined Fast-HotStuff liveness are presented in the Appendix.
\end{proof}

\begin{figure}[ht!]
\centering      
\subfigure[Latency for 1MB Block Size]{      
\begin{minipage}{4cm}
\centering                                          \includegraphics[width=1.8in]{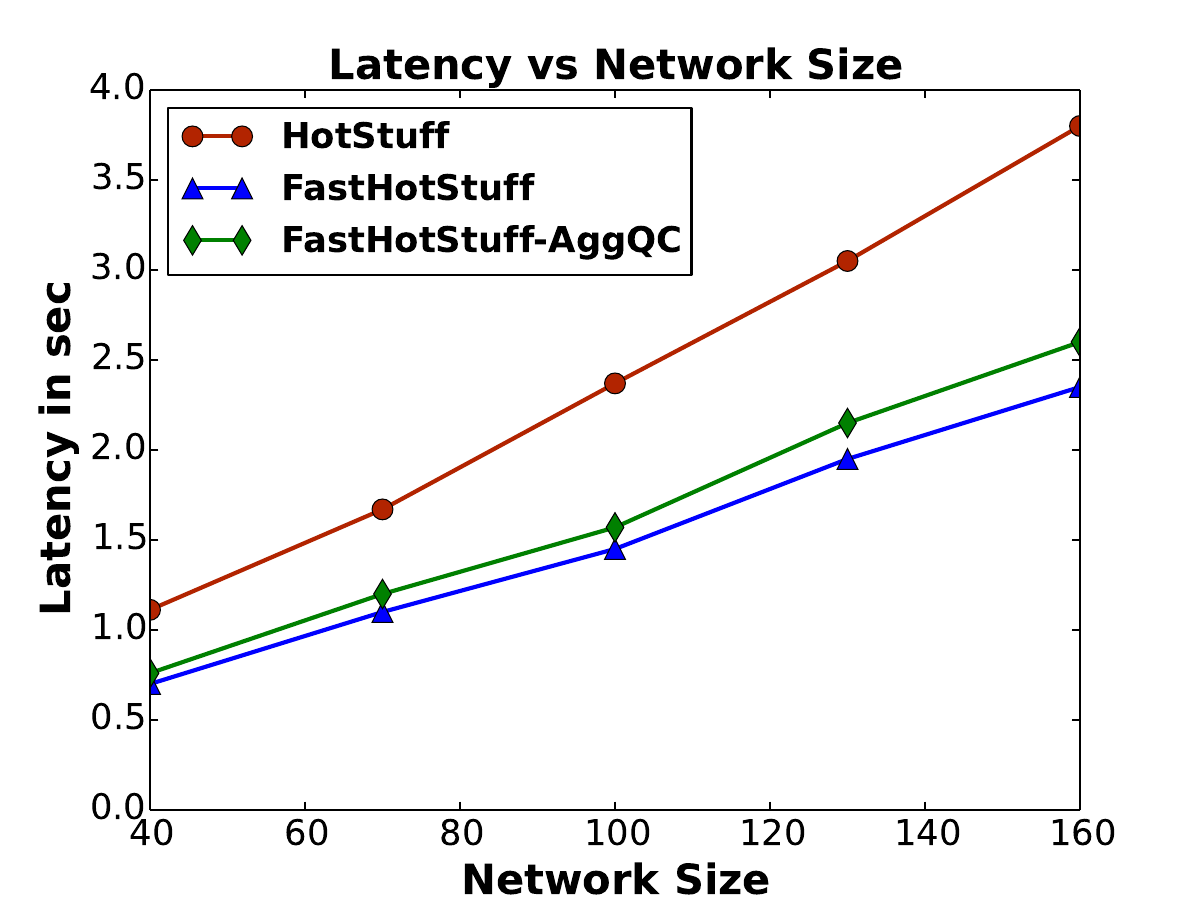}  
\label{fig:FHS-Latency-1MB}   
\end{minipage}
}\hfill
\subfigure[Latency for 2MB Block Size]{    
\begin{minipage}{4cm}
\centering                                              \includegraphics[width=1.8in]{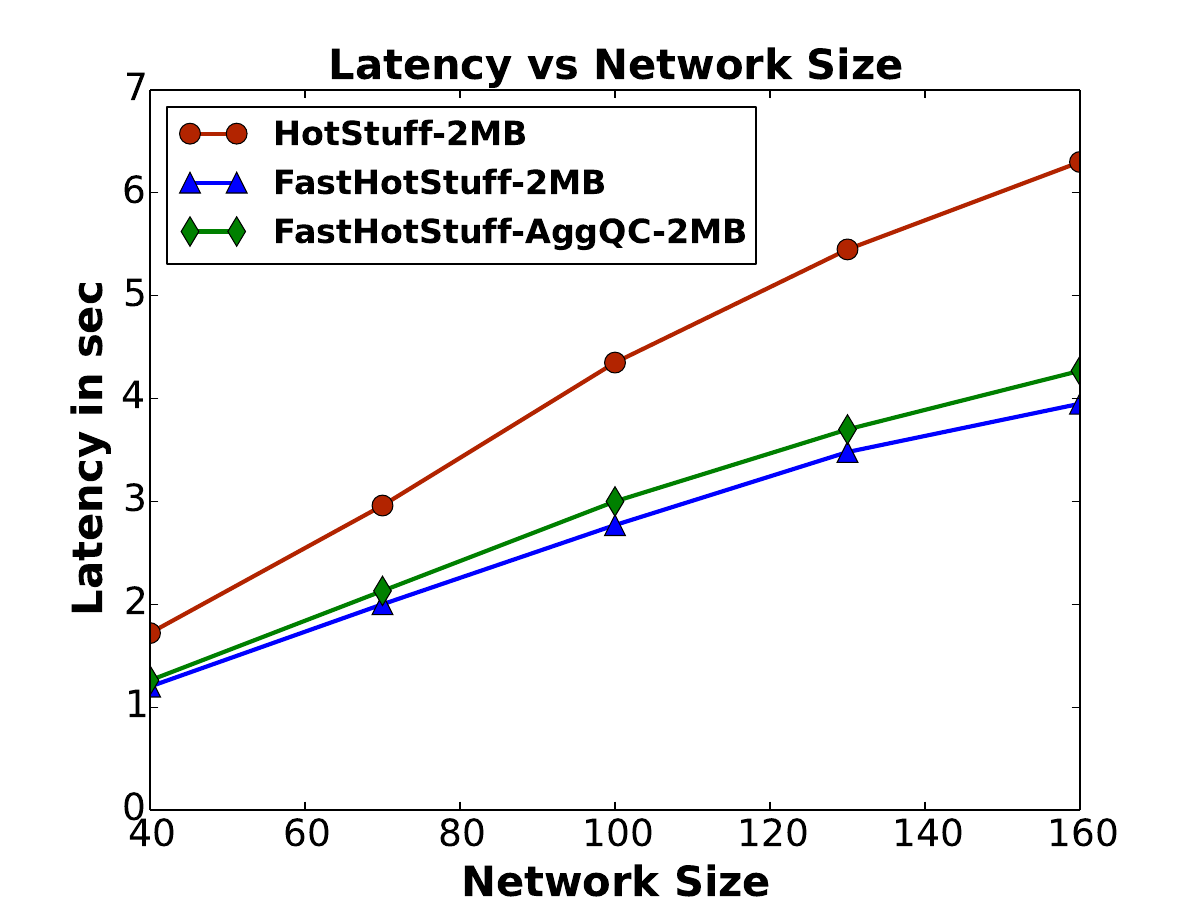}   
\label{fig:FHS-Latency-2MB} 
\end{minipage}
}\hfill
\subfigure[The Latency under the Forking Attack for 1MB Block Size.]{    
\begin{minipage}{4cm}
\centering                                                          \includegraphics[width=1.8in]{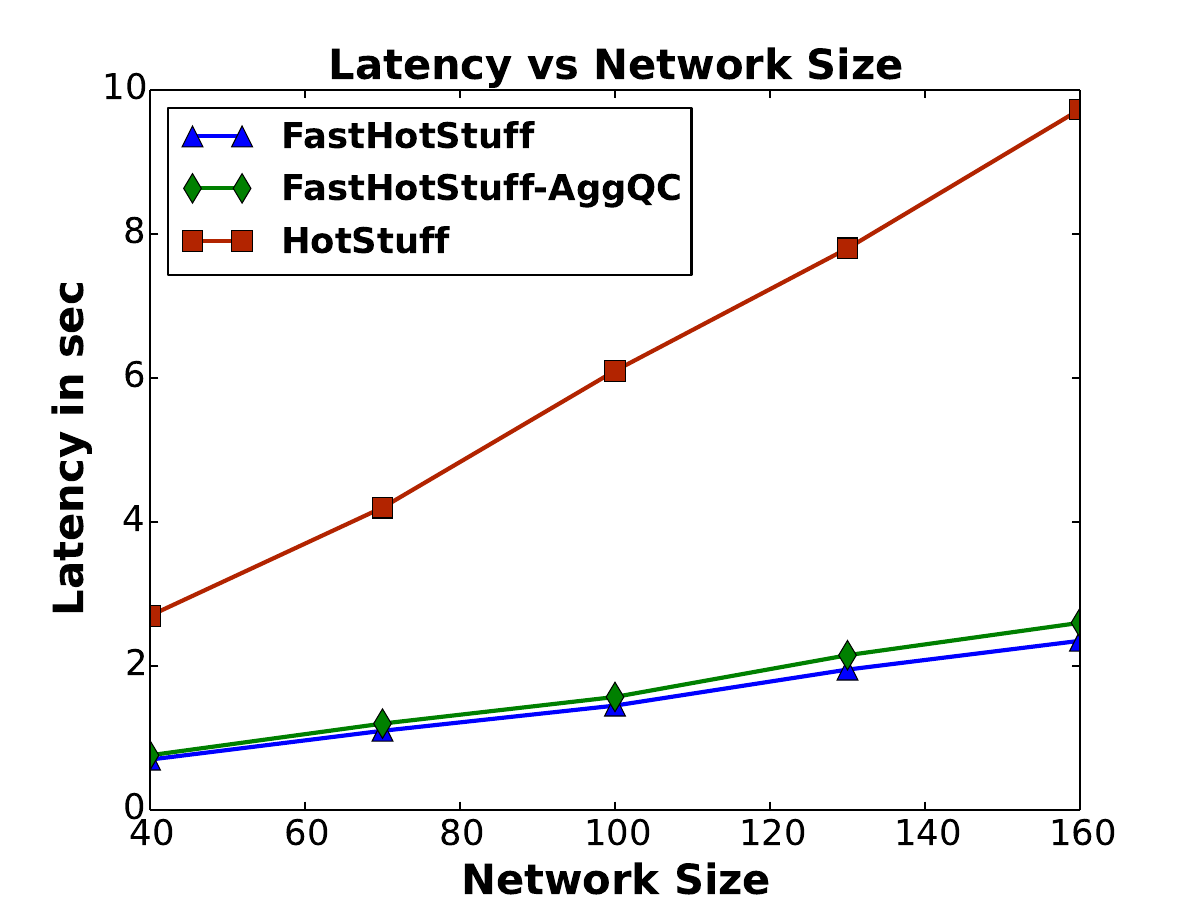}   
\label{fig:FHS-Latency-1MB-UA}   
\end{minipage}
}
\subfigure[The Latency under the Forking Attack for 2MB Block Size.]{    
\begin{minipage}{4cm}
\centering                                                          \includegraphics[width=1.8in]{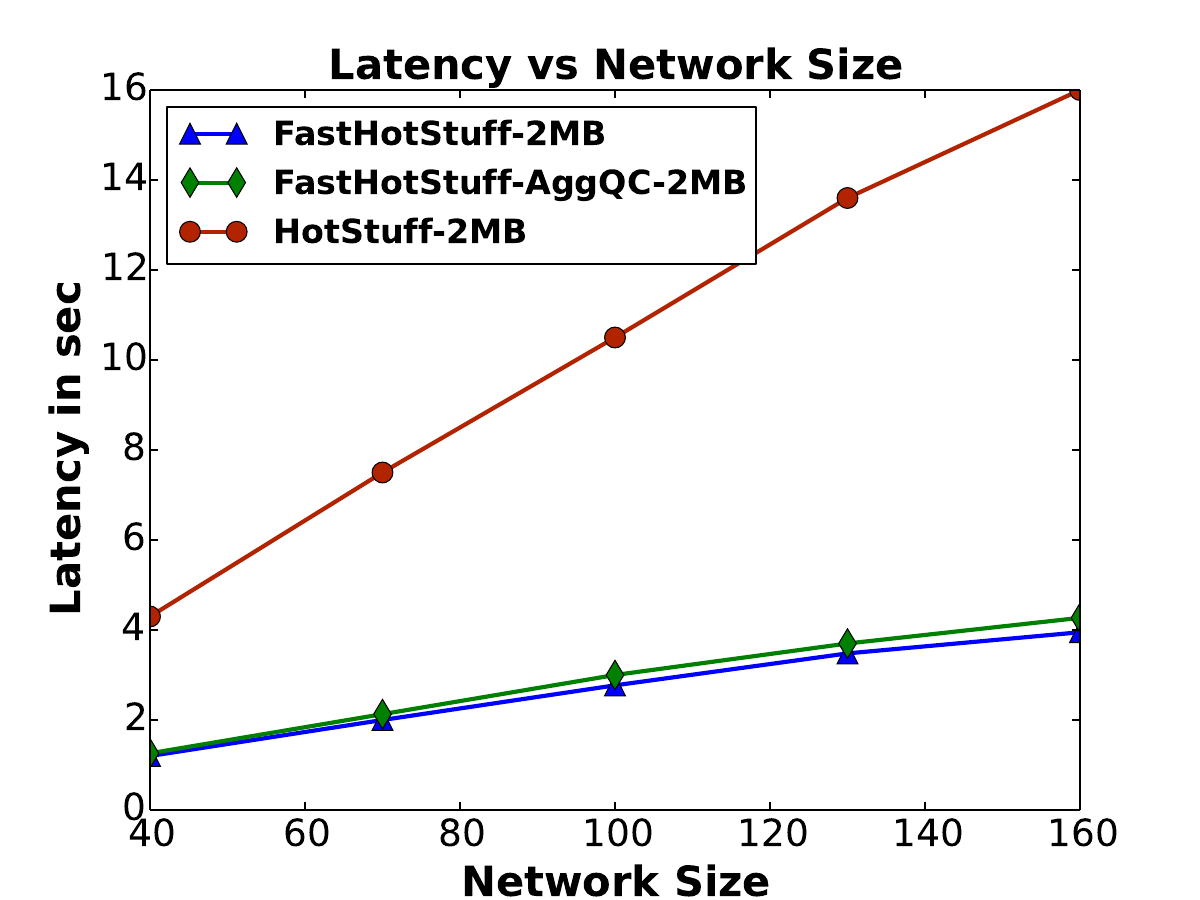}   
\label{fig:FHS-Latency-2MB-UA}   
\end{minipage}
}

\caption{The latency comparison of HotStuff and Fast-HotStuff in LAN environment.}  
\label{fig:LatencyQualityQC}                                 
                   
\end{figure}

\begin{figure}[t]
\centering   
\subfigure[Throughput for 1MB Block Size]{      
\begin{minipage}{4cm}
\centering                                          \includegraphics[width=1.8in]{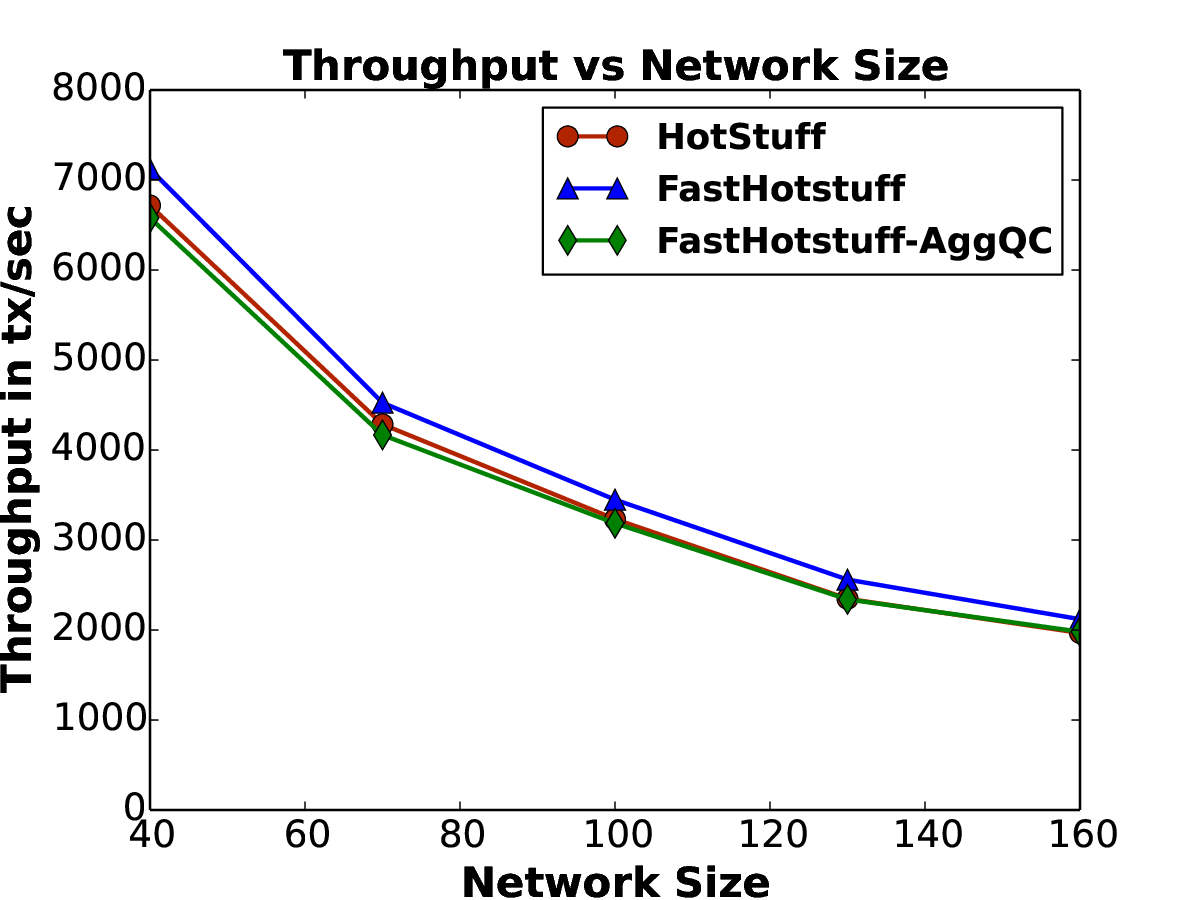}   
\label{fig:FHS-Throughput-1MB}   
\end{minipage}
}\hfill
\subfigure[Throughput for 2MB Block Size]{    
\begin{minipage}{4cm}
\centering                                              \includegraphics[width=1.8in]{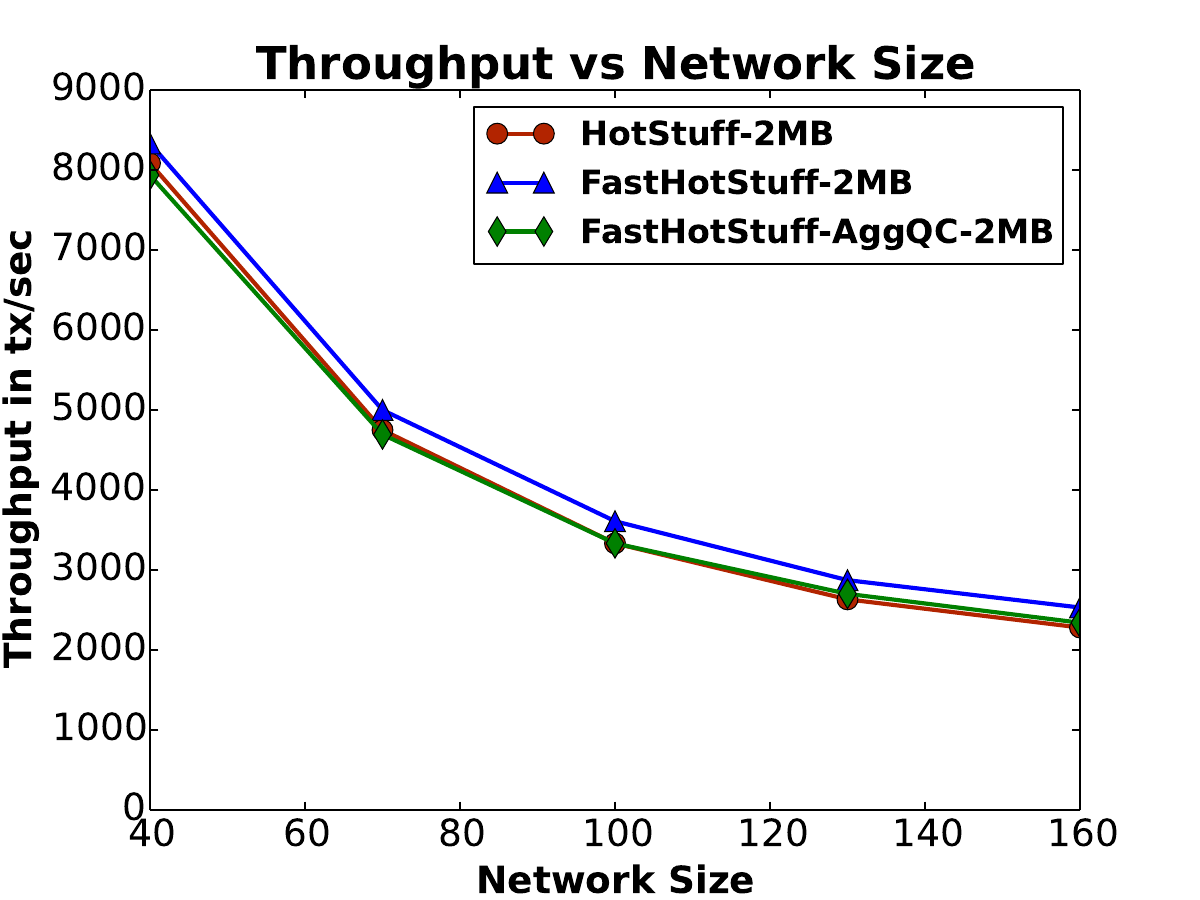}   
\label{fig:FHS-Throughput-2MB} 
\end{minipage}
}\hfill
\subfigure[The Throughput under the Forking Attack for 1MB Block Size.]{    
\begin{minipage}{4cm}
\centering                                                          \includegraphics[width=1.8in]{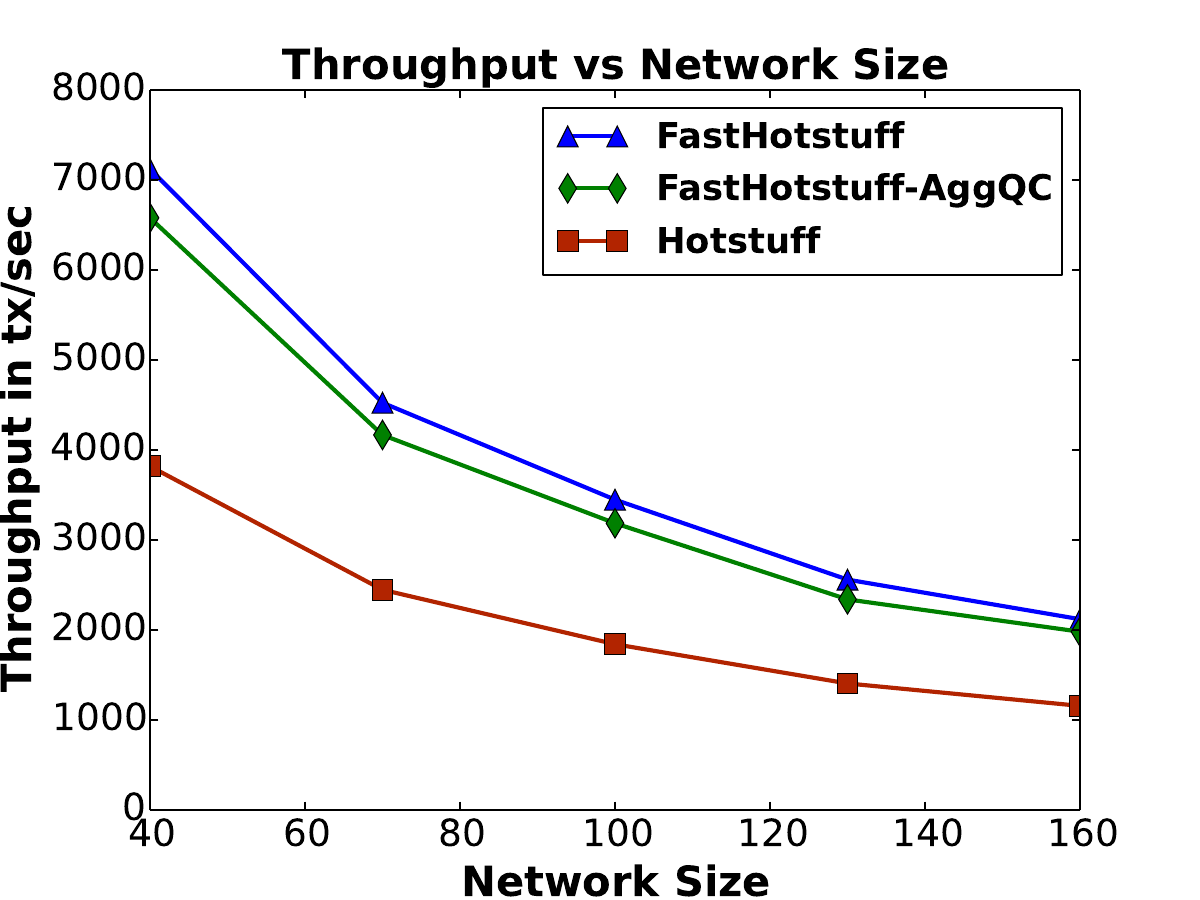}   
\label{fig:FHS-Throughput-1MB-UA}   
\end{minipage}
}\hfill
\subfigure[The Throughput under the Forking Attack for 2MB Block Size.]{    
\begin{minipage}{4cm}
\centering                                                          \includegraphics[width=1.8in]{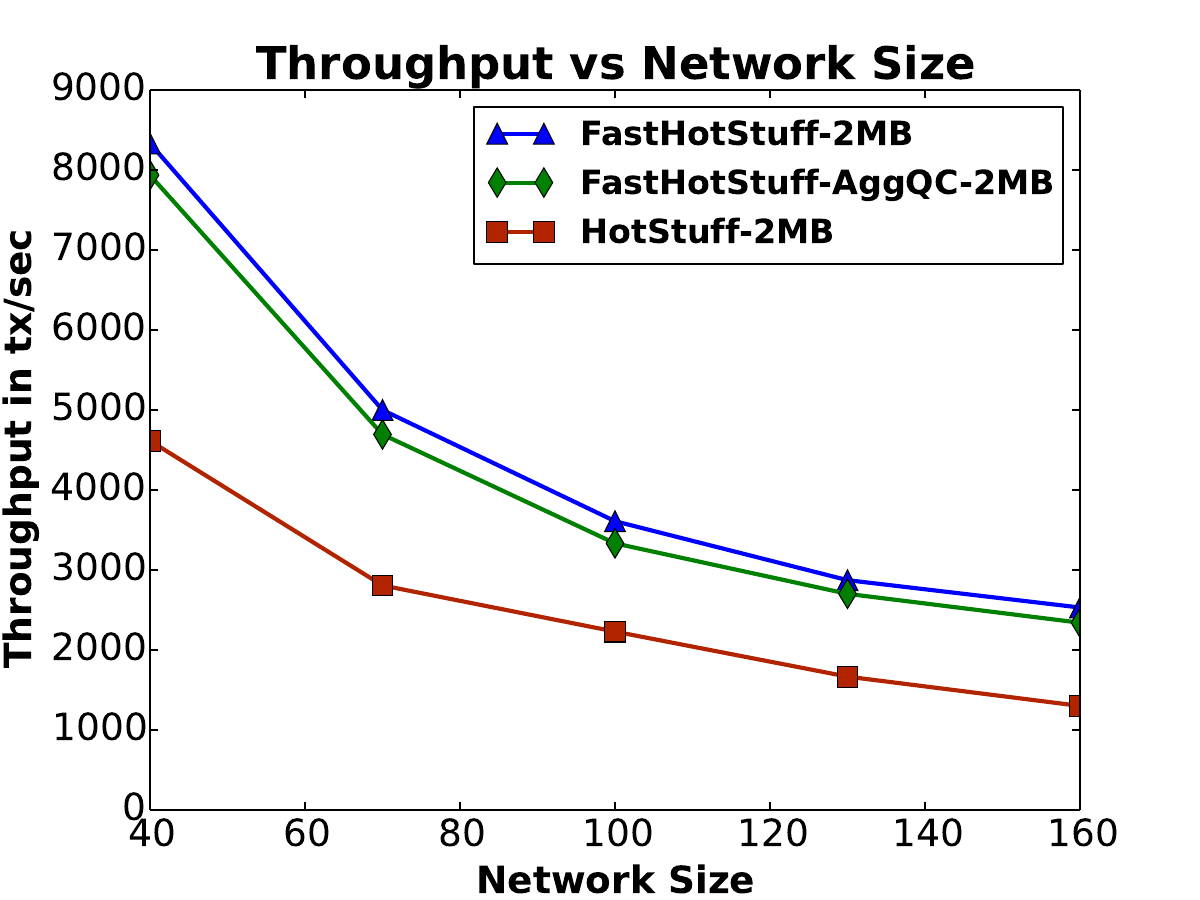}   
\label{fig:FHS-Throughput-2MB-UA}   
\end{minipage}
}

\caption{The Throughput comparison of HotStuff and Fast-HotStuff in LAN environment.}  
\label{fig:growthQualityQC}                                                        
\end{figure}


\begin{figure}
    \centering
    \includegraphics[scale=0.3]{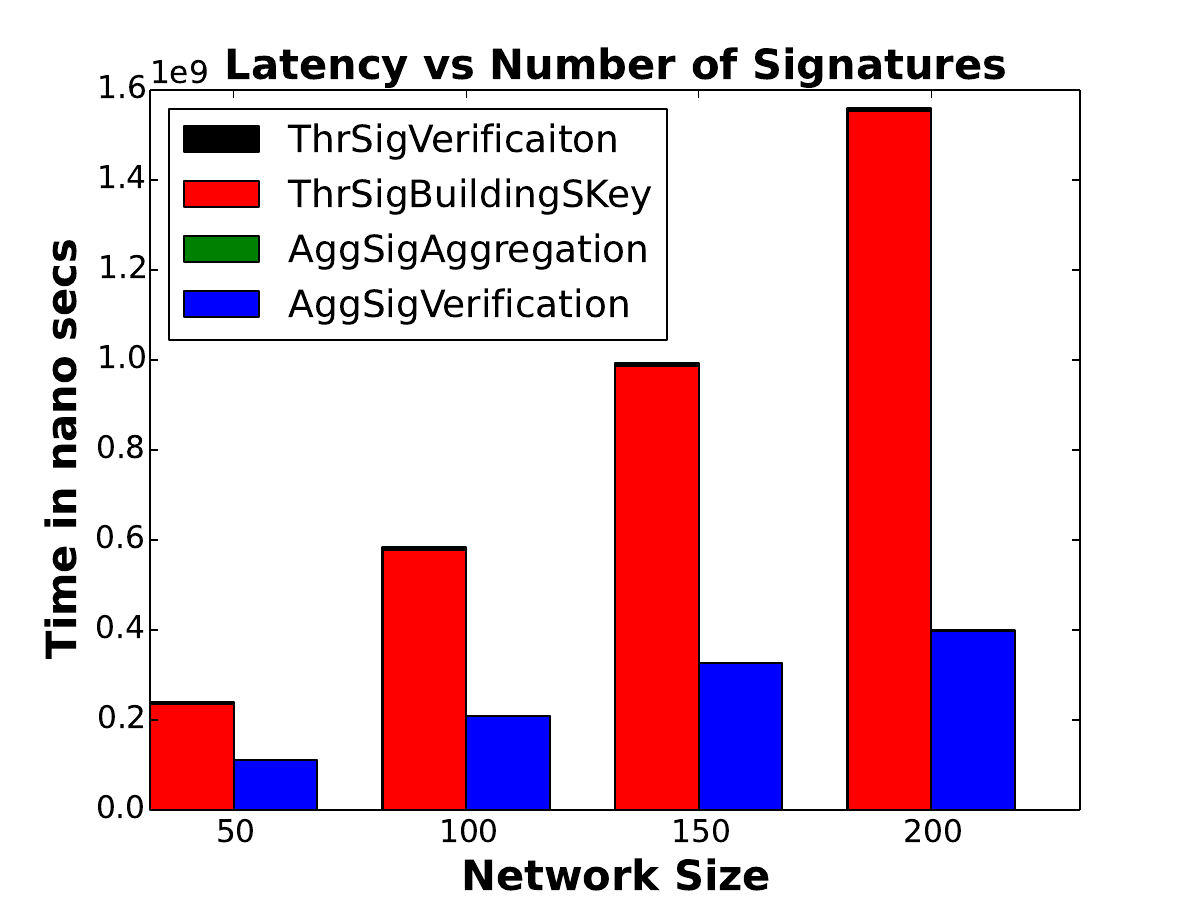}
    \caption{\textbf{Comparing threshold and Aggregated Signature Latency}.}
    \label{fig:Signature-Schemes}
\end{figure}


\begin{figure}[t]
\centering      
\subfigure[Latency for 1MB Block Size]{      
\begin{minipage}{4cm}
\centering                                          \includegraphics[width=1.8in]{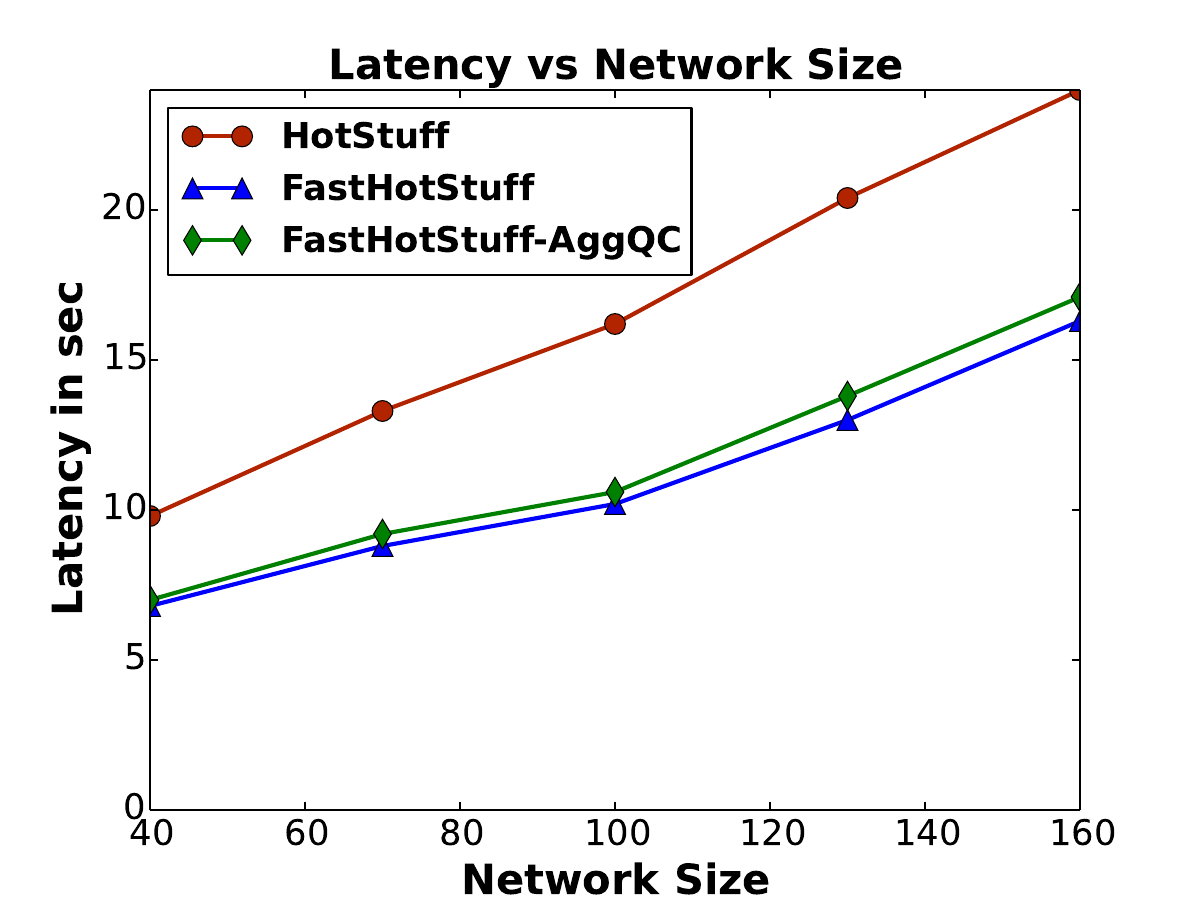}  
\label{fig:FHS-Latency-WAN-1MB}   
\end{minipage}
}\hfill
\subfigure[Latency for 2MB Block Size]{    
\begin{minipage}{4cm}
\centering                                              \includegraphics[width=1.8in]{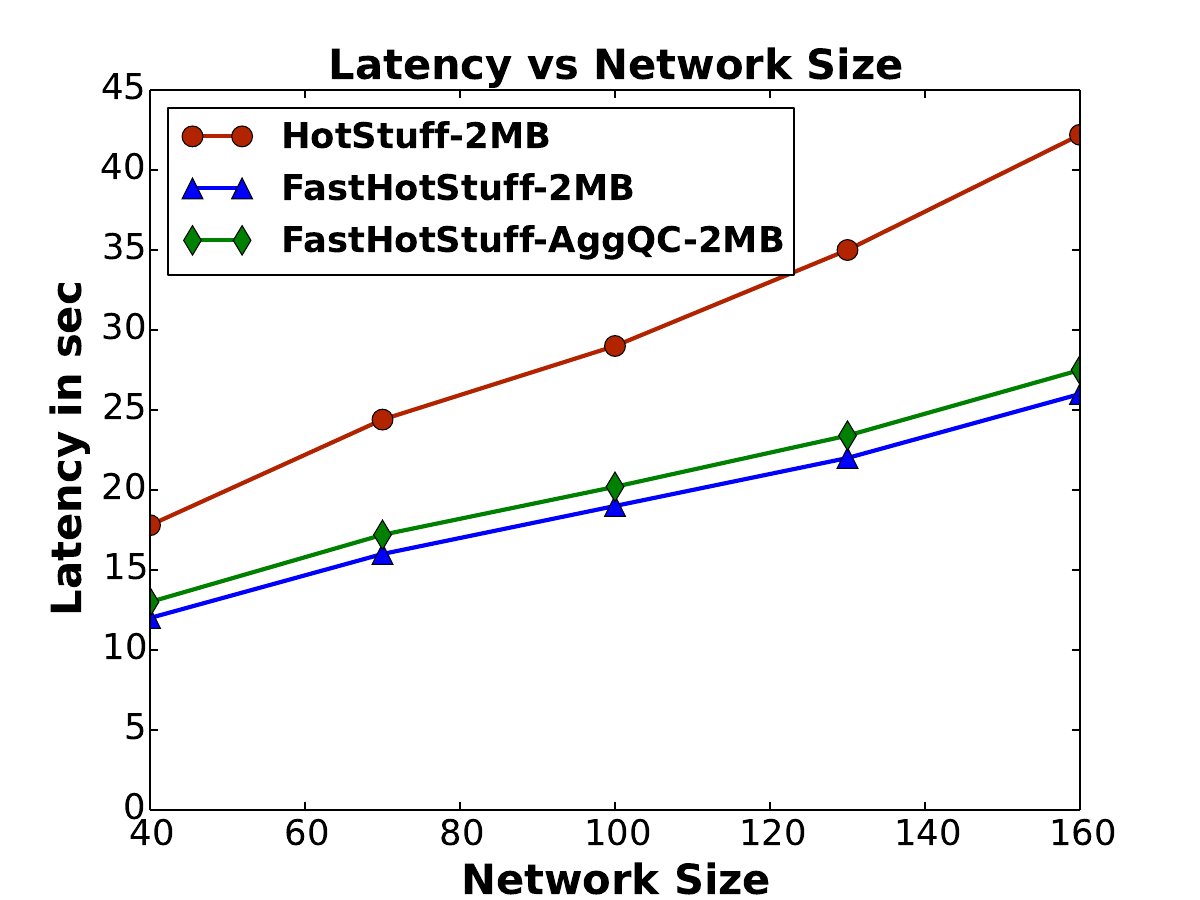}   
\label{fig:FHS-Latency-WAN-2MB} 
\end{minipage}
}\hfill
\subfigure[The Latency under the Forking Attack for 1MB Block Size.]{    
\begin{minipage}{4cm}
\centering                                                          \includegraphics[width=1.8in]{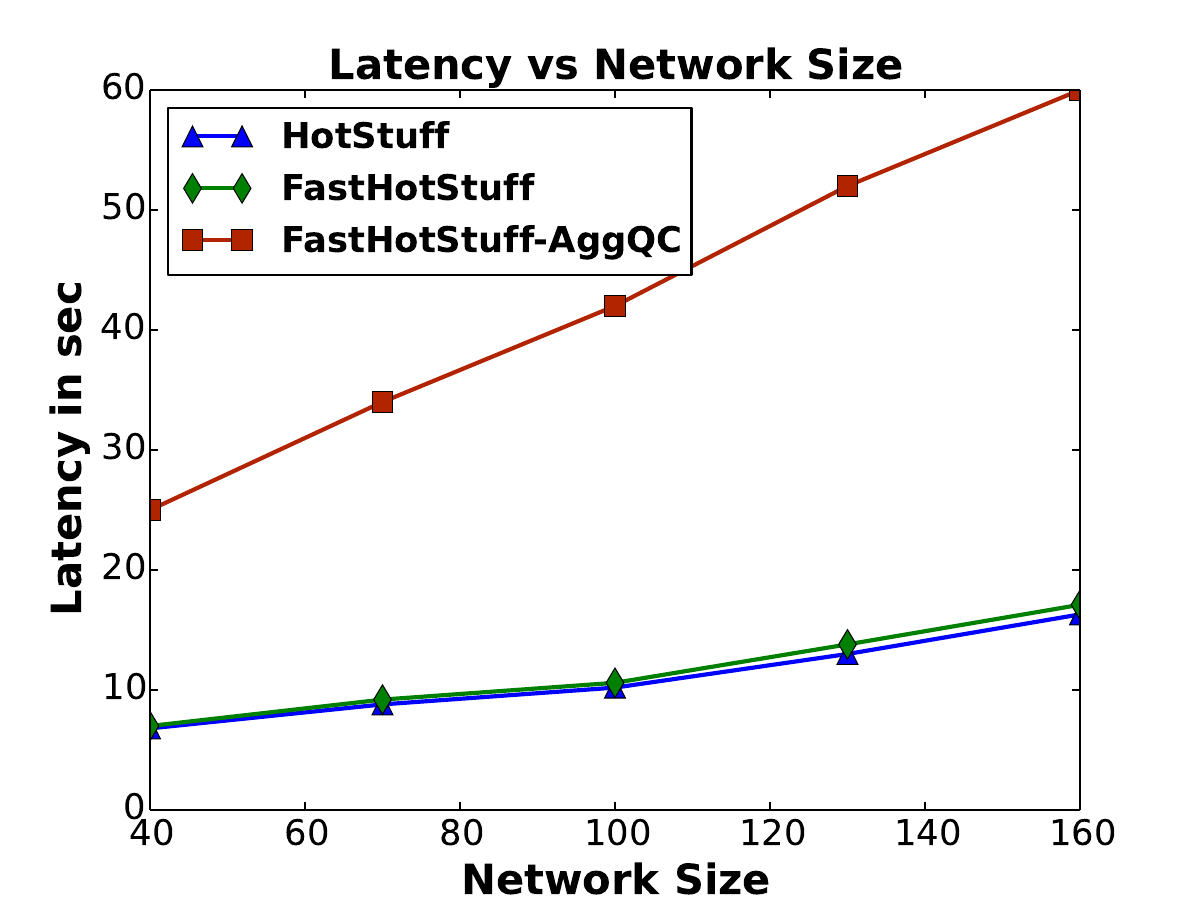}  
\label{fig:FHS-Latency-WAN-1MB-UA}   
\end{minipage}
}\hfill
\subfigure[The Latency under the Forking Attack for 2MB Block Size.]{    
\begin{minipage}{4cm}
\centering                                                          \includegraphics[width=1.8in]{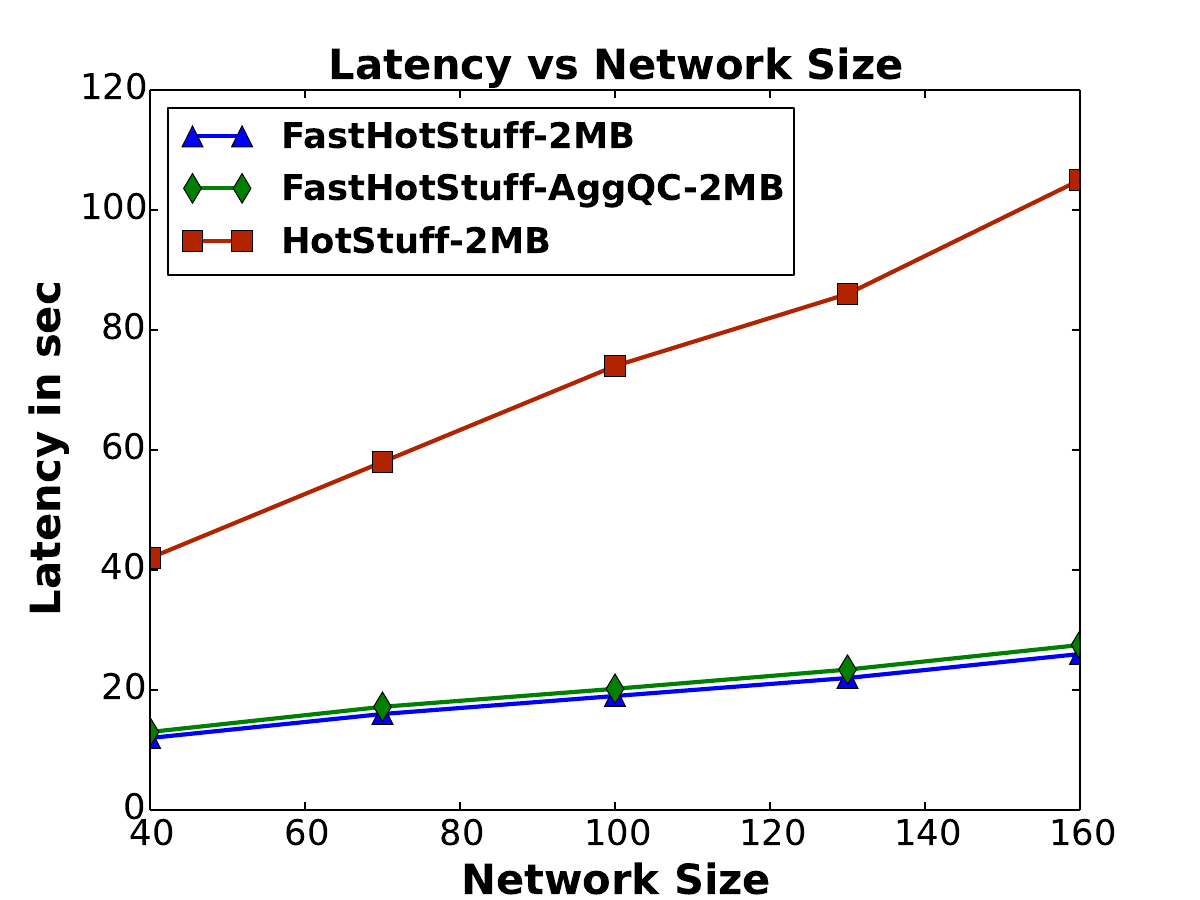}  
\label{fig:FHS-Latency-WAN-2MB-UA}   
\end{minipage}\hfill
}

\caption{The latency comparison of HotStuff and Fast-HotStuff in WAN environment.}  
\label{fig:LatencyQualityQC}                                 
                   
\end{figure}


\begin{figure}[t]
\subfigure[Throughput for 1MB Block Size]{      
\begin{minipage}{4cm}
\includegraphics[width=1.8in]{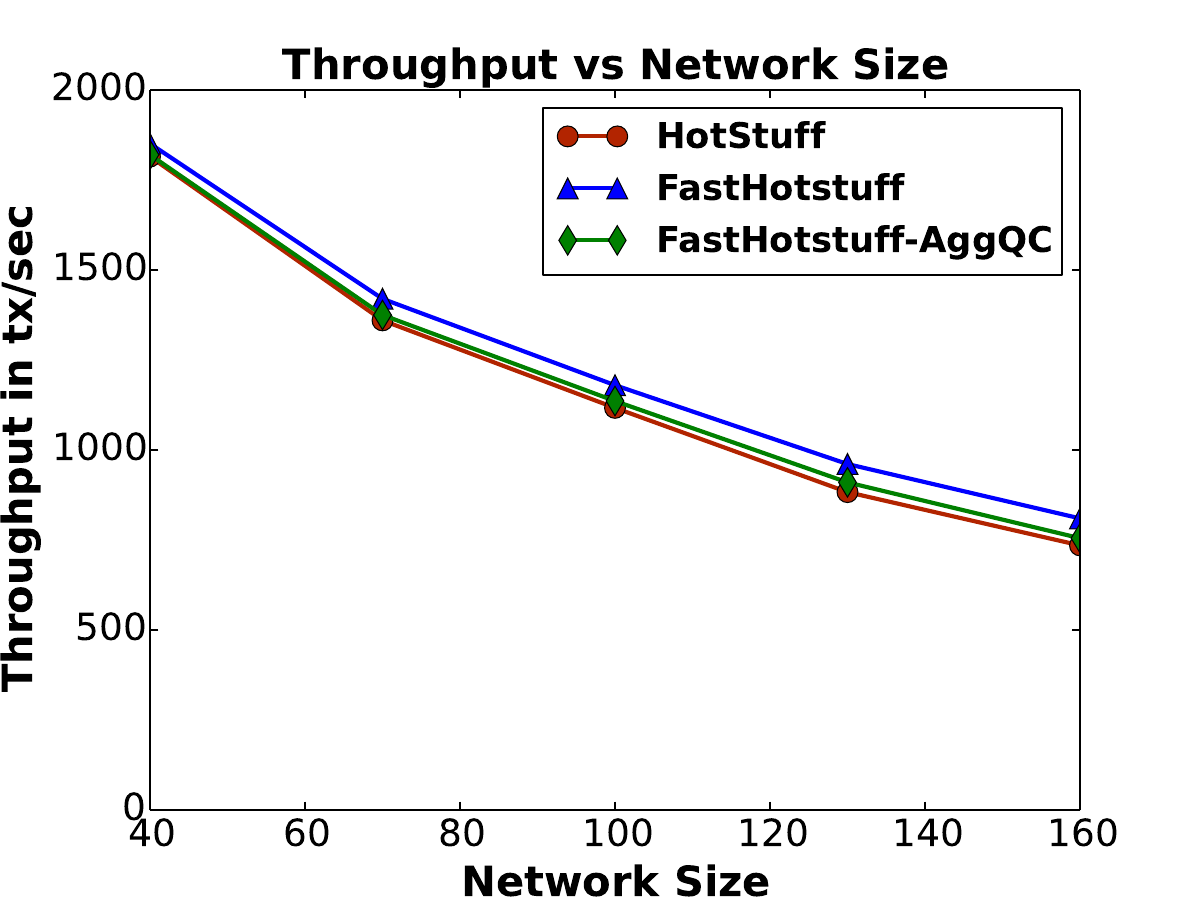}   
\label{fig:FHS-Throughput-WAN-1MB}   
\end{minipage}
}\hfill
\subfigure[Throughput for 2MB Block Size]{    
\begin{minipage}{4cm}
\includegraphics[width=1.8in]{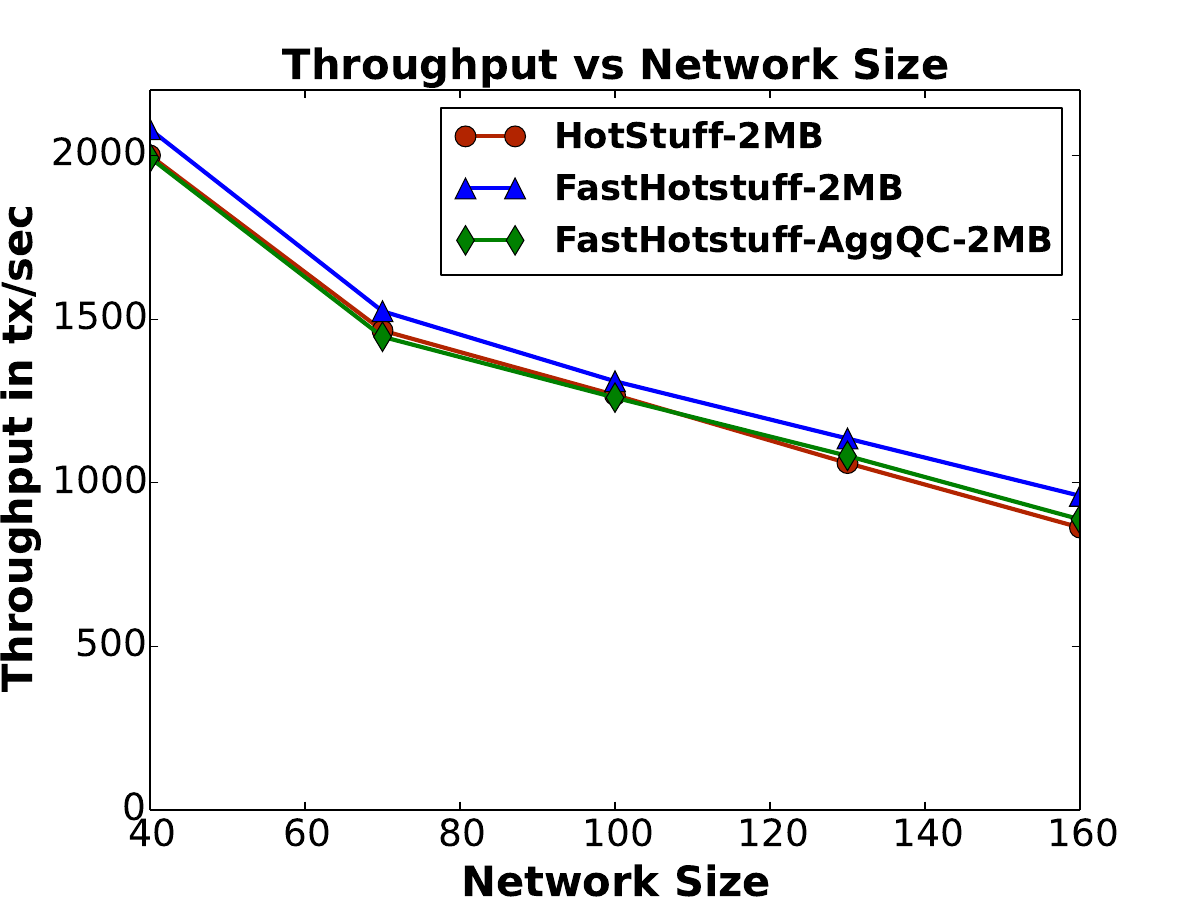}   
\label{fig:FHS-Throughput-WAN-2MB} 
\end{minipage}
}\hfill
\subfigure[The Throughput under the Forking Attack for 1MB Block Size.]{    
\begin{minipage}{4cm}
\includegraphics[width=1.8in]{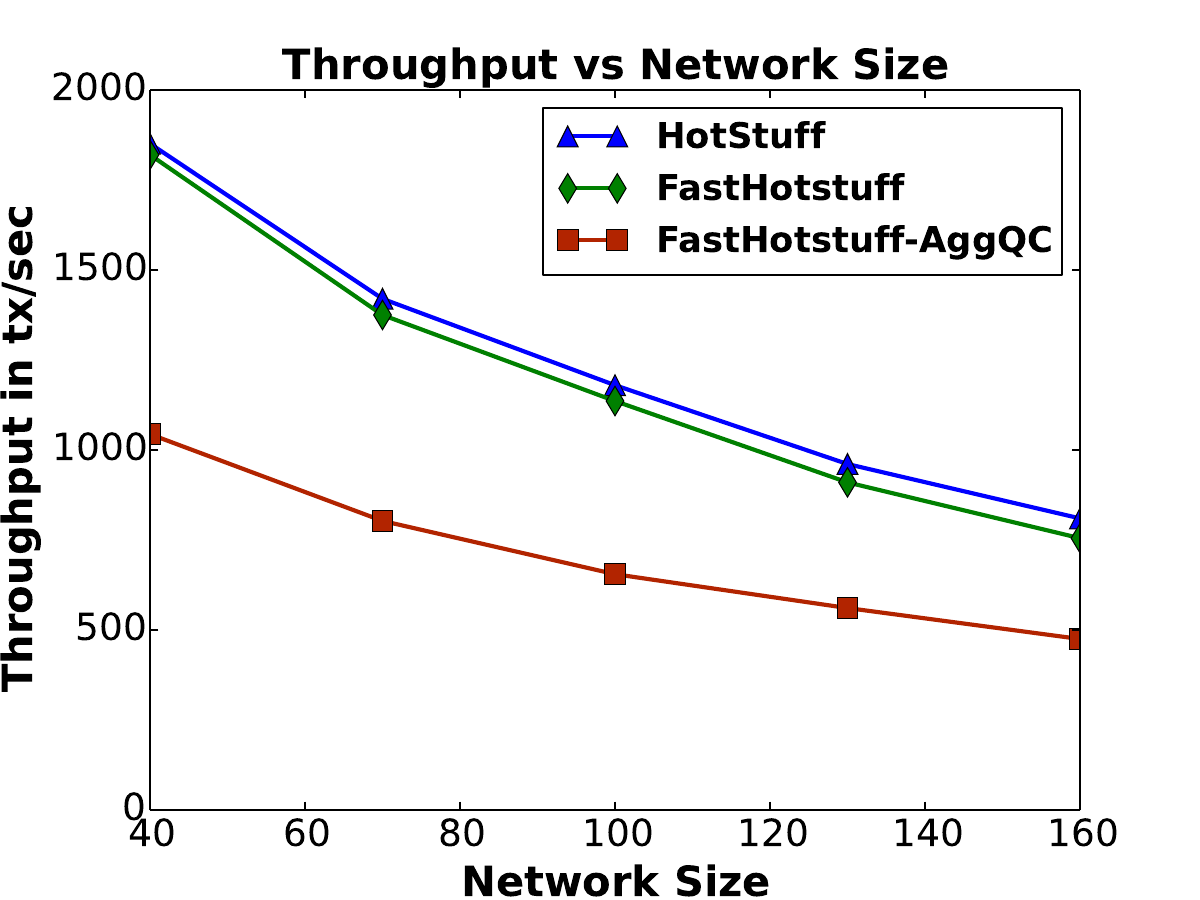}   
\label{fig:FHS-Throughput-WAN-1MB-UA}   
\end{minipage}
}\hfill
\subfigure[The Throughput under the Forking Attack for 2MB Block Size.]{    
\begin{minipage}{4cm}
\centering                                                          \includegraphics[width=1.8in]{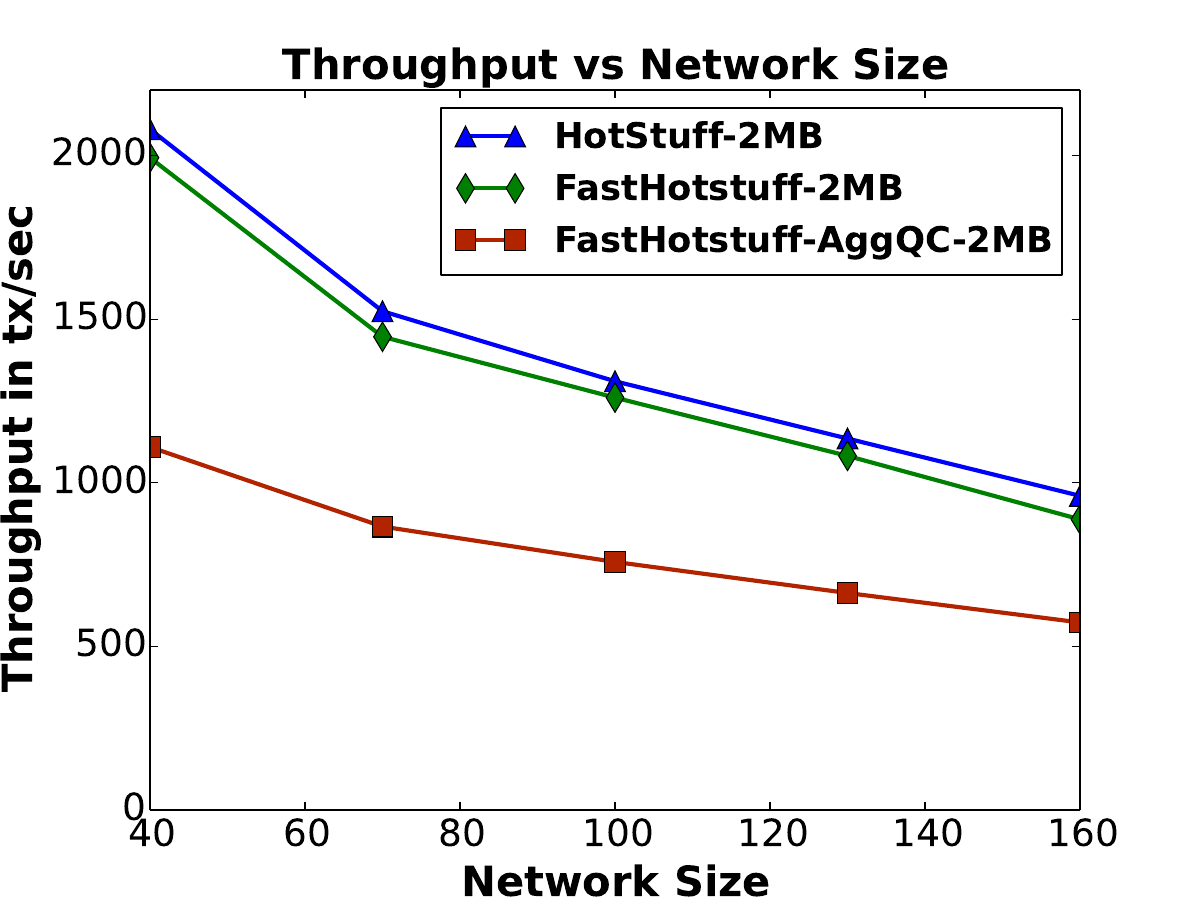}   
\label{fig:FHS-Throughput-WAN-2MB-UA}   
\end{minipage}
}

\caption{The Throughput comparison of HotStuff and Fast-HotStuff in WAN environment.}  
\label{fig:WAN-Throughput}                                                        
\end{figure}


\section{Evaluation}
\label{Section: Evaluation}
As pipelined HotStuff is widely adopted due to its higher throughput, we implemented  prototypes of pipelined Fast-HotStuff and  pipelined HotStuff using \textcolor{black}{the} Go programming language. For cryptographic operations, we used dedis advance\fangyu{d} crypto library in Go \cite{dedis-Signatures}. \textcolor{black}{Dedis library supports both BLS aggregated signatures as well as threshold signatures (used during evaluation) \cite{Short-Signatures-from-the-Weil-Pairing}.} Moreover, for hashing values\fangyu{,} we used SHA256 hashing \chen{in} the Go crypto library~\footnote{https://golang.org/pkg/crypto/sha256/}. We tested Fast-HotStuff's performance on Amazon cloud (AWS) with different network sizes  $40,70,100,130,160$ (to compare performance and scalability) and different block sizes ($1MB$ and $2MB$). We used $t2.micro$ replicas in AWS, where each replica has a single vCPU (virtual CPU) comparable to a single core with $1GB$ memory. The bandwidth was set to $500$ Mb/sec (62.5 MB/sec) and the latency between two endpoints was set to $50$ msec.  We also performed forking a attack on each network of different sizes,  $k=1000$ times, and took \textcolor{black}{the} mean of the average throughput of the network and the latency incurred by blocks affected by this attack.

We compared pipelined Fast-HotStuff's performance (throughput and latency) against pipelined HotStuff in three scenarios: 1) during happy path (when no failure occurs)\fangyu{,} therefore\fangyu{,} the primary only includes the $QC$ in the block (red vs blue curves), 2) when the previous primary fails and the next primary in the Fast-HotStuff has to include the aggregated $QC$ (($n-f$) $QC$s) in the block (red vs green curves), and 3) when HotStuff and Fast-HotStuff come under forking attack. Results achieved from each case \chen{are} discussed below:
\begin{enumerate}
    \item As the results shown in Figure \ref{fig:FHS-Latency-1MB} and \ref{fig:FHS-Latency-2MB}, during happy path when no failure occurs\chen{,} Fast-HotStuff \chen{outperforms} HotStuff in terms of latency. HotStuff's throughput (in red) slightly decreases against Fast-HotStuffs throughput (in blue) due to $O(n^2)$ time complexity for interpolation calculation at the primary (for threshold signatures in HotStuff) when $n$ increases (Figure \ref{fig:FHS-Throughput-1MB} and \ref{fig:FHS-Throughput-2MB}). 
    \item We did not consider the timeout period taken by replicas to recover\chen{,} as our main objective is to show that even the inclusion of aggregated $QC$ (small overhead) after a primary failure has a negligible effect on pipelined Fast-HotStuff's performance. 
Therefore, just after GST (after primary failure), when an honest (correct) primary is selected, HotStuff does not incur additional overhead. \fangyu{However,} HotStuff needs an additional round of consensus during happy (normal) as well as unhappy (failure) cases in comparison to Fast-HotStuff. As a result, the throughput and latency of HotStuff during the happy path or unhappy path just after primary failure is the same and is shown by the red curve (by ignoring the timeout period).  
Though, Fast-HotStuff's throughput slightly decreases, due to the $AggQC$ (a small overhead) as shown in \ref{fig:FHS-Throughput-1MB} and \ref{fig:FHS-Throughput-2MB}, 
but the throughput is still comparable to the HotStuff throughput.
On the other hand, Fast-HotStuff's latency (in green) is lower than HotStuff's latency (red) as shown in Figure \ref{fig:FHS-Latency-1MB} and Figure \ref{fig:FHS-Latency-2MB}.

\item Forking attack effects in Fast-HotStuff and HotStuff can be observed in Figure \ref{fig:FHS-Latency-1MB-UA} and \ref{fig:FHS-Latency-2MB-UA} and Figure \ref{fig:FHS-Throughput-1MB-UA} and \ref{fig:FHS-Throughput-2MB-UA}.
\textcolor{black}{Forking attack (the use of older $QC$ in a block) is allowed in HotStuff, but since forking  is not allowed in Fast-HotStuff, the proposal will be rejected and the proposer will not be chosen again.} 
 Hence, forking attack has no significant effect on Fast-Hotstuf performance. 
\end{enumerate}

\label{SubSection:Performance attacks}
\textbf{Choice of Signatures}
 HotStuff uses \textit{t-out-of-n} threshold signatures \cite{Short-Signatures-from-the-Weil-Pairing} and the Fast-HotStuff uses aggregated signatures as stated previously. 
 To better understand the effects of  signature schemes on each protocol's performance, we also compared the latency caused by each signature type during consensus. Higher latency caused by cryptographic operations during consensus will result in higher latency and lower throughput.

In both aggregated signature (used in Fast-HotStuff) as well as threshold signature (used in HotStuff) there are three main steps for signature verification in HotStuff and Fast-HotStuff. First, each replica signs a message (vote, $NEWVIEW$) and sends it to the primary. Second, the primary build an aggregated signature or threshold signature from $n-f$ messages received from distinct replicas and sends it back to each replica. Third, each replica verifies the aggregated or threshold signature associated with the messages received from the primary.
The cost for the first step is small and constant (for constant size message) for both threshold signature as well as aggregated signature. Therefore, we ignore this cost. The computation cost of aggregating signatures into a single constant size aggregated signature by the primary replica in Fast-HotStuff  is small, resulting in a negligible delay.
It is shown by the green bar stacked over top of the blue bar in Figure \ref{fig:Signature-Schemes}. But since the time taken by signature aggregation in Fast-HotStuff's primary is so small, it is hardly visible in the graph.
However, the cost of building the threshold signature from $n-f$ partial signatures received by the primary replica in HotStuff from $n-f$ replicas (shown by the red bar) in the threshold signatures in HotStuff is quadratic as it uses $O(t^2)$ ($t=n-f$) time polynomial interpolation. 
Conversely, the computation cost of verification in threshold signature is small per replica, shown by the black bar stacked over the red bar in Figure \ref{fig:Signature-Schemes}.
However, the cost of verifying aggregated signature in the aggregated signature scheme is linear, resulting in a linear delay shown by the blue bar.
In summary, we can see that building threshold signatures in the primary from partial signatures is expensive in threshold signatures. Although threshold signatures have constant verification costs, the latency induced in the primary prevents threshold signatures from having  performance gain in comparison to the aggregated signature scheme.

\textbf{WAN Performance} \textcolor{black}{
Furthermore, to better understand the effect of high latency over the Wide Area Network (WAN), we did performance tests over the Wide Area Network (WAN). Nodes are arranged in three different AWS regions including US-East-1 (North Virginia) region, EU-Central-1 (Frankfurt), and US-East-1 (North Virginia). The latency between the US-East-1 (North Virginia) region and EU-Central-1 (Frankfurt) was $88ms$ , US-East-1 (North Virginia) and South-America-East-1 (Sao-Paulo) was $140ms$ and EU-Central-1 (Frankfurt), and South-America-East-1 (Sao-Paulo) was $227ms$.
It can be seen both protocols experience lower throughput over the WAN in comparison to the throughput within a single data center.}

\textcolor{black}{ Similarly, in the WAN setting, Fast-HotStuff has lower latency than the HotStuff (as shown in Figure \ref{fig:FHS-Latency-WAN-1MB} and \ref{fig:FHS-Latency-WAN-2MB}). During the normal case, a block of 1MB in HotStuff will take at most $8$ seconds (approximately) more than a 1MB block in Fast-HotStuff (with network size $160$). This means clients have to wait for an additional $8$ seconds to get transaction confirmation. But the situation can get worse if a forking attack is performed by Byzantine nodes (shown in the Figure \ref{fig:FHS-Latency-WAN-1MB-UA}). In such a case, on average a single block can take additional $40$ seconds to get confirmed. This additional block confirmation time in HotStuff can reach more than $16$ seconds during normal cases and more than $70$ seconds when a forking attack is performed with a 2MB block size  (with network size $160$) as shown in Figure \ref{fig:FHS-Latency-WAN-2MB-UA}. This high latency is due to the reason that in chained HotStuff four blocks (with at least the first three with consecutive views also called two-direct chain) have to be added for the first block (among the four) to get committed. But during the forking attack, the chain of blocks from consecutive views breaks, hence the network has to wait for the completion of the \emph{two-direct chain} condition to meet so that block/s can get committed. Similarly, HotStuff's throughput in WAN is comparable with the Fast-HotStuff's throughput but significantly degrades under  \textcolor{black}{a} forking attack (as shown in Figure \ref{fig:WAN-Throughput}). 
}

\section{Related Work}\label{Section: Related Work}
\begin{figure*}
    \captionof{table}{Comparison of BFT Protocols}
        \centering

\addtolength{\tabcolsep}{-0.4em}\begin{tabular*}{\textwidth}{c@{\hspace{1.5\tabcolsep}}cccccc}

\hline
Protocol & \multicolumn{1}{p{2cm}}{\centering Authenticators Sent During View Change} & \multicolumn{1}{p{3cm}}{\centering Optimistic Responsiveness }& \multicolumn{1}{p{2cm}}{\centering \# of Rounds During Happy Path} & \multicolumn{1}{p{2cm}}{\centering Leader Paradigm }& \multicolumn{1}{p{2cm}}{\centering Authenticators Verified  during View Change} &  \multicolumn{1}{p{2cm}}{\centering \# of Rounds from View Change to receiving first block} \\
\hline
 PBFT\cite{Castro:1999:PBF:296806.296824} & $O(n^2)$ or $O(n^3)$  & Yes & 2 & stable & $O(n^2)$ or $O(n^3)$  &   2\\ 
SBFT\cite{SBFT} & $O(n^2)$ & Yes & 1-2 & stable & $O(n^2)$ &  2 \\ 
 Tendermint \cite{tendermint} & $O(n)$ &No & 2 & Rotating & $O(n)$ &  1\\
 Casper FFG \cite{Casper} & $O(n)$ & No & 2 & Rotating & $O(n)$  & 1\\
 
 HotStuff\cite{Hot-stuff} & $O(n)$& Yes & 3 & Rotating & $O(n)$ &  1\\
 LibraBFT (DiemBFT) & $O(n)$ & Yes & 3 & Rotating & $O(n)$ &  1\\
 Fast-HotStuff  & $O(n)$ or $O(n^2)$   & Yes & 2& Rotating & $O(n)$ &  1\\
\hline
\label{Table: Protocol-Comparision}
\end{tabular*}
\end{figure*}

There have been multiple works on improving BFT protocols performance and scalability \cite{Jalal-Window,SBFT,BFT-SMART}. But these protocols suffer from expensive view changes  where either the message complexity or the number of signatures/authenticators to be verified grows quadratic\fangyu{ly}. Moreover, these protocols do not employ a primary rotation mechanism.  \textcolor{black}{Current solutions (see Table \ref{Table: Protocol-Comparision}) suffer from at least one of the shortfalls below.}

\textbf{Lack of optimistic responsiveness} 
Protocols that offer a simple mechanism of leader/primary replacement include Casper \cite{Casper} and Tendermint \cite{tendermint}. But both of these protocols have synchronous cores, where, replicas in the network have to wait for \fangyu{the} maximum network latency before moving to the next round. In other words, these protocols lack responsiveness, which will result in high-performance degradation. 
The BFT consensus protocol proposed in \cite{OnOptimality-BFT-CCS} 
is responsive if \textcolor{black}{the} number of Byzantine replicas is less than $n/4$ ($f<n/4$). The protocol losses responsiveness and waits for the maximum network delay, when $f \geq n/4$. 
Fast-HotStuff is always responsive regardless of the value of $f$ and can tolerate $f < n/3$ Byzantine replicas.

\textbf{Expensive View Change}
PBFT \cite{Castro:1999:PBF:296806.296824} has quadratic message complexity during happy path or normal operation. During view change, each replica has to process at least $O(n^2)$ signatures. Moreover, PBFT does not use a rotating primary. On the other hand, Fast-HotStuff uses a rotating primary, pipelined Fast-HotStuff has linear view change during normal primary rotation (view change). In case of failure, primary rotation in Fast-HotStuff in each replica has to process only two aggregated signatures.

SBFT \cite{SBFT} has linear communication complexity during normal operation. SBFT does not employ a rotating primary mechanism. When Fast-HotStuff and SBFT use the same signature scheme, replicas in Fast-HotStuff will have to verify $O(n)$ signatures less than the SBFT.

\textcolor{black}{To address the higher latency and vulnerability to the forking attack, in a concurrent research effort DiemBFT has switched to a two-chain protocol in its new release.  During the normal phase, the new DiemBFT(V4) behaves similar to the Fast-HotStuff, it employs PBFT \cite{Castro:1999:PBF:296806.296824} based quadratic view change (as an intermediary solution). Fast-HotStuff has $O(n)$ times lower signature verification cost, during the unhappy path than this variant of DiemBFT(v4).
As a further improvement, the DiemBFT is implementing the asynchronous fallback from \cite{gelashvili2021jolteon} for the view change so that the protocol can even make progress in the event of asynchrony (including Denial-of-Service attacks on the primary).}

While asynchronous view change improves resilience, it suffers from high bit complexity. \textcolor{black}{For example,} for a network of $n=100$ replicas with the block size as $b_s=1MB$, the total number of bits sent by the primary is $n \times b_s + n \times 1.4 \times \frac{b_s}{100}=100 MB + n \times  \frac{1.4 MB} {100} = 101.4MB$ (block plus overhead sent to $n$ replicas). The added overhead in this case in Fast-HotStuff is $n  \times 1.4 MB /100 = 1.4MB$. 
For a proposal in \cite{gelashvili2021jolteon} since each replica has to broadcast a block of size $b_s$, therefore $n$ replicas broadcast $n$ blocks out of which only one block will eventually be added to the chain. In this case, we have $n \times n \times b_s = 100 \times 100 \times 1MB = 10^4 MB$ complexity. Therefore, $10^4$ $MB$ data have to be exchanged for a single block of $1MB$ to be committed (by each one of $100$ replicas). 
Hence, only $100$ MB of data is being consumed (committed) by the network. The additional $10^4-100=9900MB$  is overhead. Therefore, all resources involved in the transmission and processing of this large amount of data are wasted. Hence, in the above-mentioned setup, the view change overhead in Fast-HotStuff, though quadratic, is $1.4 MB$  whereas, in \cite{gelashvili2021jolteon} is $9900MB$. \textcolor{black}{ Similarly, in the Fast-HotStuff, signature verification cost during view change is $O(n)$ times lower than this variant of DiemBFT(v4).} As a result, deployment of this protocol \cite{gelashvili2021jolteon} \textcolor{black}{(as a variant of DiemBFT(v4))} will require high resources including bandwidth, processor (as each transaction within block and signatures have to be verified), and memory (for keeping blocks until only one chain is selected). Moreover, values of $n$ and $b_s$ have to remain in check.

\textbf{Suboptimal latency}
HotStuff \cite{Hot-stuff} is designed to not only keep a simple leader change process but also maintain responsiveness. These features along with \fangyu{the} pipelining optimization have provided an opportunity for wide adoption of the HotStuff \cite{libra-BFT} \fangyu{protocol}. HotStuff has linear view change, but we show that in practice during primary failure (unhappy path) Fast-HotStuff's view change performance is comparable to HotStuff. LibraBFT (DiemBFT) \cite{libra-BFT} is a variant of HotStuff. Unlike the HotStuff which uses threshold signatures, LibraBFT uses aggregated signatures. 
LibraBFT uses broadcast\fangyu{ing} during primary failure.
LibraBFT also has a three-chain structure and is susceptible to forking attacks. 




\textbf{View change latency} \textcolor{black}{ View change latency cannot be bounded during the asynchronous period. Therefore, we consider the number of rounds when a replica intends to change the view to the first time it receives a proposal from a correct primary (during the period of synchrony). This period can determine how quickly the protocol recovers when synchrony conditions meet in the presence of a correct primary.
It takes two rounds by a replica running PBFT or SBFT to receive the first proposal while having a successful view change.
} 

\textcolor{black}{Another beautiful concurrent work (Wendy) in  \cite{No-Commit-Proofs} has also addressed the higher latency during the happy path in HotStuff while supporting rotating primary.  Wendy has successfully avoided the quadratic authenticator complexity during the view change while maintaining only two rounds of latency during the happy path.  The tradeoff Wendy has made is that if the failed primary or the new (selected) primary (when less than  $f+1$ correct nodes are holding locks) is Byzantine, then it can  force the protocol to incur additional latency costs during the view change (in terms of the number of rounds). 
Furthermore, the paper also assumes that there is a known fixed bound on the difference among the views of replicas in the network. On the other hand, Fast-HotStuff is resilient against such performance attacks and does not have such an assumption. By contrast, the tradeoff Fast-HotStuff has made is that a primary has to send $O(n^2)$ authenticators after a view change. And we showed that the bandwidth cost of this overhead is small in practice and improved  the signature verification of the protocol by $O(n)$.} 

\textcolor{black}{The verification  cost of authenticators depends on the type of signature schemes used by the protocol. 
For example, the verification cost of  aggregated signatures is linear to the number of signatures being aggregated. Whereas, the verification cost of threshold signatures is constant. But on the other hand, as previously observed during experiments, the quadratic computational cost of Lagrange interpolation adds additional latency to the protocol. Lagrange interpolation is used  
by the primary to build a threshold signature from partial signatures. As a result, the constant size and verification cost of threshold signatures do not benefit the protocol performance.  Though a recent work \cite{Scalable-Threshold} has tried to address the high computational cost of interpolation in threshold signatures with a tradeoff in verification and communication cost.
}

\section{Conclusion}
\label{Section:Conclusion}
In this paper\fangyu{,} we presented \textcolor{black}{Fast-HotStuff} which is a two-chain consensus protocol with efficient and simplified view change. It achieves consensus in two rounds of communication. Moreover, Fast-HotStuff is robust against forking attacks. Whereas HotStuff lacks resilience against forking attacks.  Fast-HotStuff achieves these unique advantages by adding a small amount of overhead in the block. This overhead is only required in rare situations when a primary fails. Our experimental results show that whether the overhead is included in the proposed block or not, Fast-HotStuff outperforms HotStuff in terms of latency. Fast-HotStuff outperforms HotStuff in terms of latency and throughput under forking attacks.


\bibliographystyle{IEEEtran}
\bibliography{new}

 \title{Appendix}
{\maketitle}
\label{Appendix}
\renewcommand{\thesubsection}{\Alph{subsection}}
\section*{Appendix}

\subsection{Performance Penalty of Direct Chain Condition.} The direct chain has a high cost in HotStuff protocol due to the forking attack. Each time when a Byzantine primary over-rides blocks for honest primaries it not only reduces the throughput but also prevents direct chain formation. This results in increased latency. On the other hand, since a forking attack is not possible in Fast-HotStuff, a direct parent cannot be broken through forking. It is also possible (in HotStuff as well as Fast-HotStuff) that a direct chain is broken when a primary fails to propose a block during a specific view  ($v$) and replicas timeout. In such a case, not only the failed primary does not propose any block but it also does not distribute the $QC$ for the block from the previous view ($v-1$). For example, as it can be seen in Figure \ref{fig:view-change}, the primary in the view $v$ has failed. Therefore, it could not disseminate the block for the view $v$ along with the $QC$ for the view $v-1$. As a result, the next primary, $v+1$, will have the $QC$ for the view $v-2$ as the $highestQC$. Hence, the block proposed in the view $v+1$ will have the $QC$ from $v-2$. Byzantine primaries can exploit the rotating primary behavior to timeout regularly and also avert one block from the previous view. 
A primary that times out frequently can be blacklisted as done in \cite{Spinning-BFT}. There can be up to $f$ number of blacklisted primaries. If an additional primary needs to be added then, the first primary from the blacklist queue is removed. In this way, the negative effect of primaries that deliberately timeout on performance can be significantly reduced.

 \subsection{Primary Selection and Liveness}
 In addition to the other generic BFT requirements for the liveness, in HotStuff and its variants, it is also important to make sure that the direct chain conditions meet eventually. If the protocol is not able to fulfill direct chain condition, then even if an honest primary is in-charge and the network is synchronous, liveness may not be guaranteed. Because no block will be committed and hence, clients will not receive response about successful execution of their transactions. Therefore, below we provide additional component of liveness proof for Fast-HotStuff under different primary selection mechanisms.
 
 \textbf{Random Primary Selection:}
 In this case the primary is chosen randomly and uniformly from a set of $N$ replicas.
 The probability of an honest primary selected is approximately $(\frac{2}{3})$. The probability of a bad event where two consecutive honest primaries not selected  is $P_b=1-(\frac{2}{3})\times (\frac{2}{3})$. By treating such a bad even as the Bernoulli trial, we have that the bad event is triggered $k$ consecutive times with probability at most $P_b^k$. As it can be seen $P_b^k$ approaches fast to $0$ at an exponential rate as $k$ grows linearly.  Moreover, by blacklisting up-to $f$ primaries for frequent failure (as stated previously), a one-direct chain can be achieved more quickly and frequently.
 Therefore, in case of random primary selection one-chain will be formed eventually.
 

\begin{figure}
    \centering
    \includegraphics[width=9cm,height=2.5cm]{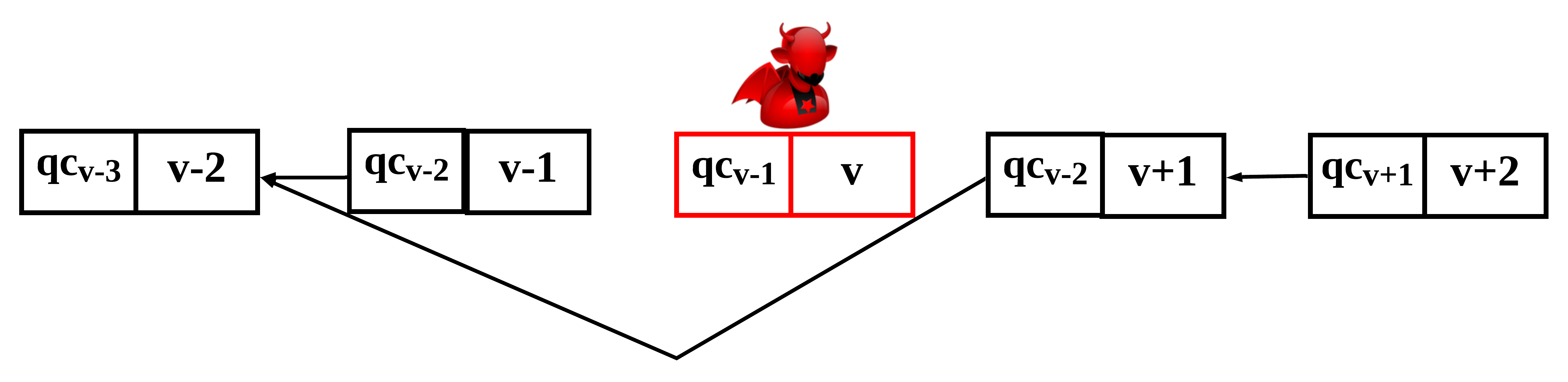}
    \caption{\textbf{Breaking one-chain through timeout by a Byzantine primary}.}
    \label{fig:view-change}
\end{figure}
 \textbf{Round-robin Primary} 
 During happy cases when there is no primary failure we know that a Byzantine primary can't perform a forking attack. Hence, one-chain will be formed between two consecutive blocks during $T_f$.
 During the view change due to the primary failure, we try to pick the best attack strategy for Byzantine primaries to prevent the formation of one-chain between two blocks from honest primaries (to cause liveness failure). The formation of one-chain can only be prevented if at least one Byzantine replica is selected as primary after each honest primary. In other words, honest primaries should never be selected consecutively to avoid one-chain formation. The best strategy for Byzantine primaries will be to prevent the maximum number of honest primaries from one-chain formation.
 There is at most $f$ number of Byzantine replicas out of  $n=3f+1$ replicas. This means if each Byzantine replica is selected as a primary after one honest primary, then the $f$ honest primaries will not be able to make one-chain. But the remaining $f+1$ honest primaries can make one-chain.  For $f=0$, since no Byzantine replica will be selected as a primary, hence, no block proposed by an honest primary will be averted. Therefore, one-chain commit condition is satisfied. For $f \geq 1$, we have $f+1 \geq 2$, hence there will be at least two consecutive honest primaries whose blocks will form one-chain and fulfill the commit condition. Therefore, if primaries are selected in a round-robin manner Fast-HotStuff holds liveness.
\end{document}